\documentclass[11pt]{article}
\usepackage[utf8]{inputenc}
\usepackage[hidelinks]{hyperref}
\usepackage{amsmath}
\usepackage{amsthm}
\usepackage{amssymb}
\usepackage{authblk}
\usepackage{cite}
\usepackage{lipsum}
\usepackage{geometry}
\usepackage{enumerate}
\usepackage{array}
\usepackage{multicol}
\usepackage{multirow}
\usepackage{xcolor}
\usepackage{kotex}
\usepackage{makecell}
\usepackage{booktabs}
\newtheorem{theorem}{Theorem}[section]
\numberwithin{theorem}{section}
\newtheorem{lemma}[theorem]{Lemma}
\newtheorem{defi}[theorem]{Definition}
\newtheorem{corollary}[theorem]{Corollary}
\newtheorem{prop}[theorem]{Proposition}
\newtheorem{remark}[theorem]{Remark}

\newcommand{\Max}{\displaystyle \max}
\newcommand{\Sum}{\displaystyle \sum}
\newcommand{\Frac}{\displaystyle \frac}
\newcommand{\F}{\mathbb F}

\newcommand{\Fcn}{\mathbb{F}_{3^n}}

\newcommand{\Fp}{\mathbb{F}_{p}}

\newcommand{\Fpn}{\mathbb{F}_{p^n}}

\newcommand{\Fcnmul}{\mathbb{F}_{3^n}^*}
\newcommand{\Fpnmul}{\mathbb{F}_{p^n}^*}

\begin{document}

\title{On APN Exponents and the Differential and Boomerang Properties of Binomials in Characteristic 3}
\author{ Namhun Koo$^1$, Soonhak Kwon$^{2,3}$, Minwoo Ko$^2$, Byunguk Kim$^2$\\
	\small{\texttt{ Email: komaton@skku.edu, shkwon@skku.edu, minwoo1403@skku.edu, kbu0923@g.skku.edu}}\\
	\small{$^1$Institute of Basic Science, Sungkyunkwan University, Suwon, Korea}\\
	\small{$^2$Department of Mathematics, Sungkyunkwan University, Suwon, Korea}\\
	\small{$^3$Applied Algebra and Optimization Research Center, Sungkyunkwan University, Suwon, Korea}
}

\maketitle

\begin{abstract}
		Recent studies on binomials of the form $F_r(x) = x^r(1 + \chi(x))$ over $\Fpn$ have shown that these functions can exhibit very low boomerang uniformity. In this paper, we focus on the specific behavior of such binomials in characteristic $3$, where instances of extremely low boomerang uniformity—namely $0$ or $1$—seem to arise more frequently than in other characteristics.
		First, we provide a systematic analysis of Almost Perfect Nonlinear (APN) power functions in characteristic $3$. We present an explicit parametrization of APN exponents arising from the construction of Zha and Wang and demonstrate through numerical results for $n \le 13$ that this generalized framework accounts for several previously known and sporadic APN instances. Building on this classification, we identify and rigorously prove two classes of binomials $F_r$ that are locally-PN and possess the minimum possible boomerang uniformity of $0$. These classes involve exponents derived from the aforementioned APN construction and the differentially 4-uniform exponent $r = 2 \cdot 3^{\frac{n-1}{2}} + 1$. Furthermore, we analyze the binomial $F_r$ with $r = 3^n - 3$, proving that it is locally-APN with boomerang uniformity $1$ when $n\ge 5$ is odd, and completely determine its boomerang spectrum through the evaluation of character sums. Our results clarify and extend existing studies on the cryptographic properties of binomials, providing a systematic characterization of several classes of binomials with very low boomerang uniformity in characteristic $3$.
	
	
\end{abstract}

\section{Introduction}

Let $p$ be an odd prime and $n$ a positive integer. We denote by $\Fpn$ the finite field with $p^n$ elements and by $\Fpnmul=\Fpn\setminus\{0\}$ its multiplicative group.
Differential and boomerang properties of functions over finite fields play a central role in the design of cryptographic primitives, as they measure resistance to differential and boomerang attacks. In particular, differential uniformity and boomerang uniformity have been extensively studied as fundamental criteria for evaluating the security of S-boxes.
Differential uniformity was introduced by Nyberg~\cite{Nyb94} and is defined as follows.

\begin{defi}
	Let $F:\Fpn \to \Fpn$ be a function. For $a\in \Fpnmul$ and $b\in \Fpn$, define $\delta_F(a,b)$ by the number of solutions of $F(x+a)-F(x)=b$. Then, the \textbf{differential uniformity} of $F$ is defined by:
	\begin{equation*}
		\delta_F = \Max_{a\in \Fpnmul, b\in \Fpn}\delta_F(a,b).
	\end{equation*} 
\end{defi}

A function $F$ is called \textbf{perfect nonlinear (PN)} if $\delta_F=1$, and \textbf{almost perfect nonlinear (APN)} if $\delta_F=2$.

While differential uniformity captures resistance against differential attacks, it does not fully reflect resistance to more advanced attacks such as boomerang attacks. To address this limitation, the notion of boomerang uniformity was originally introduced for permutations by Cid et al.~\cite{CHP+18}. This notion was later extended to general (not necessarily bijective) functions by Li et al.~\cite{LQSL19}, leading to the following definition.

\begin{defi}
	Let $F:\Fpn \to \Fpn$ be a function. For $a,b \in \Fpnmul$, define $\beta_F(a,b)$ as the number of pairs $(x,y)\in \Fpn^2$ satisfying:
	\begin{equation*}
		\begin{cases}
			F(x)-F(y)=b,\\
			F(x+a)-F(y+a)=b.
		\end{cases}
	\end{equation*}
	Then the \textbf{boomerang uniformity} of $F$ is given by:
	\begin{equation*}
		\beta_F = \Max_{a,b\in \Fpnmul}\beta_F(a,b).
	\end{equation*} 
\end{defi}

While differential uniformity has been extensively studied for general functions, much of the work on differential spectrum and boomerang uniformity has focused on power functions. 
Functions of the form $F_{r,u}(x)=x^r(1+u\chi(x))$, where $\chi(x)=x^{\frac{p^n-1}{2}}$ denotes the quadratic character on $\Fpn$, were first studied in \cite{NH07,ZHYJ07}, where the case $r=p^n-2$ was shown to yield APN functions in certain characteristics. These functions have also been investigated in more recent works~\cite{LWZ24,RXY26,XBC+24,XLB+25}, which study their differential properties in the non-APN cases.
In~\cite{LWZ24,MW25,KK26}, particular choices of the exponent $r$, such as $r=p^n-2,\,2$, and $\frac{p^n+1}{4}$, were further investigated. These works also considered the case $u=1$, leading to functions of the form 
$$F_r(x)=x^r(1+\chi(x)),$$ 
which were not treated in the same way in other works \cite{RXY26,XBC+24,XLB+25}. The above results show that the functions $F_r$ can exhibit very low values of $\Max_{b\in\Fpnmul}\delta_{F_r}(1,b)$ and low boomerang uniformity.

Motivated by this phenomenon, in our previous work~\cite{KKKK25} we initiated a systematic study of such functions. It was shown that if $r$ satisfies the condition that the equation
$
(x+1)^r - x^r = b
$
has at most one solution with $\chi(x)=\chi(x+1)=1$ for each $b\in \Fpnmul$, and if $\gcd(r,p^n-1)\le 2$, then the function $F_r$ satisfies
\begin{itemize}
	\item $\delta_{F_r}=\delta_{F_r}(1,0)=\frac{p^n+1}{4}$,
	\item $\Max_{b\in\Fpnmul}\delta_{F_r}(1,b)\le 2$,
	\item $\beta_{F_r}\le 2$.
\end{itemize}
Moreover, several classes of such exponents $r$ were identified. We refer the reader to~\cite{KKKK25} or Section~\ref{subsec_KKKK25} for further details. Tables~\ref{table diff_Fr} and~\ref{table_boomerang_Fr} summarize known results on the differential and boomerang properties of the functions $F_r$, together with the new classes studied in this paper.

\renewcommand{\arraystretch}{1.1}
\begin{table}[htbp]
	\centering
	\begin{tabular}{ccccc}
		\toprule
		$r$ 
		& Conditions 
		& $\max\limits_{b\in \Fpnmul}\delta_{F_r}(1,b)$
		& \makecell{Spectrum\\ provided} 
		& Ref. \\
		\midrule
		
		$p^n-2$ 
		& $p=3$, $n$ odd
		& $1$ 
		& O 
		& \cite{LWZ24} \\
		
		$p^n-2$ 
		& $p>3$, $p^n \equiv 3 \pmod{4}$ 
		& $1$ 
		& O 
		& \cite{LWZ24} \\
		
		$2$ 
		& $p^n \equiv 3 \pmod{4}$ 
		& $\le 2$ 
		& X 
		& \cite{MW25} \\

		$\frac{p^n+1}{4}$ 
		& $p^n \equiv 3 \pmod{8}$ 
		& $1$ 
		& O 
		& \cite{KK26} \\
		
		$\frac{p^n+1}{4}$ 
		& $p^n \equiv 7 \pmod{8}$ 
		& $2$ 
		& O 
		& \cite{KK26} \\
		
		$p^k+1$ 
		& $p^n \equiv 3 \pmod{4}$, $0 < k < n$   
		& $\le 2$ 
		& X 
		& \cite{KKKK25} \\
		
		$\frac{3^k+1}{2}$ 
		& $p=3$, $n$ odd, $0 < k < n$ 
		& $\le 2$ 
		& X 
		& \cite{KKKK25} \\
		
		$3$ 
		& $p^n \equiv 11 \pmod{12}$ 
		& $2$ 
		& O 
		& \cite{KKKK25} \\
		
		$\frac{2 p^n-1}{3}$ 
		& $p^n \equiv 11 \pmod{12}$ 
		& $2$ 
		& O 
		& \cite{KKKK25} \\
		
		$\frac{3^{\frac{n+1}{2}} - 1}{2}$ 
		& $p=3,\; n$ odd 
		& $\le 2$ 
		& X 
		& \cite{KKKK25} \\
		
		$\frac{3^{n+1}-1}{8}$ 
		& $p=3,\; n$ odd 
		& $\le 2$ 
		& X 
		& \cite{KKKK25} \\

		$\frac{3^{um}-1}{3^m+1}$ 
		& \makecell{$p=3$, $n$ odd, $\gcd(m,n)=1$\\
			$um\equiv 1 \pmod n$, $u$ even}
		& $1$ 
		& O
		& Section \ref{subsec_r=ZW10} \\
		
		$\frac{1-3^{(n-u)m}}{1+3^m}$ 
		& \makecell{$p=3$, $n$ odd, $\gcd(m,n)=1$\\
			$um\equiv 1 \pmod n$, $u$ odd}
		& $1$ 
		& O 
		& Section \ref{subsec_r=ZW10} \\
		
		$2 \cdot 3^{\frac{n-1}{2}} + 1$ 
		& $p=3,\; n$ odd 
		& $1$ 
		& O 
		& Section \ref{subsec_r=2*3^l+1} \\
		
		$3^n-3$ 
		& $p=3,\; n\ge 5$ odd 
		& $\le 2$ 
		& X 
		& Section \ref{sec_r=-2} \\
		
		\bottomrule
	\end{tabular}
	\caption{Differential uniformity and spectrum of $F_r$}\label{table diff_Fr}
\end{table}

\begin{table}[htbp]
	\centering
	\begin{tabular}{ccccc}
		\toprule
		$r$ 
		& Conditions 
		& $\beta_{F_r}$ 
		& \makecell{Spectrum\\ provided}  
		& Ref. \\
		\midrule
		
		$p^n-2$ 
		& $p=3$, $n$ odd
		& $0$ 
		& - 
		& \cite{LWZ24} \\
		
		$p^n-2$ 
		& $p^n \equiv 3 \pmod{4}$, $p>3$ 
		& $1$  
		& O 
		& \cite{LWZ24} \\
		
		$2$ 
		& $p^n \equiv 3 \pmod{4}$, large $p^n$ 
		& $2$ 
		& X 
		& \cite{MW25} \\
		
		$\frac{p^n+1}{4}$ 
		& $p^n \equiv 3 \pmod 8$ 
		& $0$  
		& - 
		& \cite{KK26} \\
		
		$\frac{p^n+1}{4}$ 
		& $p^n \equiv 7 \pmod 8$ 
		& $2$  
		& O 
		& \cite{KK26} \\
		
		$p^k+1$ 
		& $p^n \equiv 3 \pmod{4}$, $0 < k < n$
		& $\le 2$ 
		& X 
		& \cite{KKKK25} \\
		
		$\frac{3^k + 1}{2}$ 
		& $p=3,\; n$ odd, $0 < k < n$ 
		& $\le 2$
		& X 
		& \cite{KKKK25} \\
		
		$3$ 
		& $p^n \equiv 11 \pmod{12}$ 
		& $\le 2$ 
		& X 
		& \cite{KKKK25} \\
		
		$\frac{2\cdot p^n-1}{3}$ 
		& $p^n \equiv 11 \pmod{12}$ 
		& $\le 2$ 
		& X 
		& \cite{KKKK25} \\
		
		$\frac{3^{\frac{n+1}{2}} - 1}{2}$  
		& $p=3,\; n$ odd 
		& $\le 2$
		& X 
		& \cite{KKKK25} \\
		
		$\frac{3^{n+1}-1}{8}$ 
		& $p=3,\; n$ odd 
		& $\le 2$ 
		& X 
		& \cite{KKKK25} \\
		
		$2$ 
		& $p=3, \; n\geq 3$ odd 
		& $1$ 
		& O 
		& \cite{KKKK25} \\
		
		$\frac{3^{um}-1}{3^m+1}$ 
		& \makecell{$p=3$, $n$ odd, $\gcd(m,n)=1$\\
			$um\equiv 1 \pmod n$, $u$ even} 
		& $0$ 
		& - 
		& Section \ref{subsec_r=ZW10} \\
		
		$\frac{1-3^{(n-u)m}}{1+3^m}$ 
		& \makecell{$p=3$, $n$ odd, $\gcd(m,n)=1$\\
			$um\equiv 1 \pmod n$, $u$ odd}
		& $0$ 
		& - 
		& Section \ref{subsec_r=ZW10} \\
		
		$2 \cdot 3^{\frac{n-1}{2}} + 1$ 
		& $p=3,\; n$ odd 
		& $0$ 
		& - 
		& Section \ref{subsec_r=2*3^l+1} \\
		
		$3^n-3$ 
		& $p=3,\; n\ge 5$ odd 
		& $1$ 
		& O 
		& Section \ref{sec_r=-2} \\

		\bottomrule
	\end{tabular}
	\caption{Boomerang uniformity and spectrum of $F_r$.}\label{table_boomerang_Fr}
\end{table}
\renewcommand{\arraystretch}{1}

In particular, in characteristic $3$, instances in which the boomerang uniformity is strictly smaller than the bound $2$ from~\cite{KKKK25} appear quite frequently, which motivates a closer investigation.
To address this, we investigate exponents $r$ for which $F_r$ attains boomerang uniformity $0$ or $1$. We perform an exhaustive search for small values of $n$, and report the results in Section~\ref{sec_table}.
Moreover, we identify two classes of exponents for which $F_r$ has boomerang uniformity $0$, namely those arising from APN exponents in~\cite{ZW10} and the class $r=2\cdot 3^{\frac{n-1}{2}}+1$, for which $x^r$ is known to be differentially $4$-uniform (see~\cite{DHKM01}). We also show that $F_r$ has boomerang uniformity $1$ when $r=3^n-3$, where $x^r$ is APN~\cite{HRS99} and its differential spectrum has been studied in~\cite{XZLH20}, and determine its boomerang spectrum.

Furthermore, we provide a more detailed analysis of APN exponents from~\cite{ZW10}. In particular, we give an explicit parametrization of such exponents (Proposition~\ref{ZW10_general_prop}) and show that, in characteristic $3$, this construction accounts for all APN exponents arising from~\cite{ZW10}. Our computational results further indicate that, for $n\le 13$, all APN power functions not listed in Table~1 of~\cite{BP25} are explained by this construction.

The rest of this paper is organized as follows. Section~\ref{sec_pre} contains some preliminaries, including a summary of the main results from our previous work~\cite{KKKK25} that will be used in later sections, as well as a brief discussion of the differential spectrum of $F_r$ in the locally-PN case, where $F_r$ also has boomerang uniformity $0$. In Section~\ref{sec_APN}, we study APN exponents arising from~\cite{ZW10}. In Section~\ref{sec_locally-PN}, we present two classes of locally-PN binomials $F_r$ with boomerang uniformity $0$. We also show that $F_{3^n-3}$ has boomerang uniformity $1$ for $n\ge 5$ odd, and explicitly evaluate its boomerang spectrum in Section~\ref{sec_r=-2}.
Section~\ref{sec_table} presents our numerical results.
Finally, Section~\ref{sec_con} concludes the paper.

\section{Preliminaries}\label{sec_pre}

\subsection{Reduction Properties of $F_r$ and Locally-APN Functions}

For power functions, the differential equation can be reduced to the case $a=1$ by a simple scaling argument. Indeed, if $F(x)=x^r$, multiplying $\frac{1}{a^r}$ on both sides of $b=F(x+a)-F(x)=(x+a)^r-x^r$, we have
\begin{equation*}
	\Frac{b}{a^r}=\left(\Frac x a +1\right)^r -\left(\Frac x a\right)^r = (y+1)^r-y^r,
\end{equation*}
where $y=\Frac x a$. Hence, we have 
\begin{equation}\label{dupower_eqn}
	\delta_F(a,b)=\delta_F\left(1, \Frac{b}{a^r} \right), \text{ so }	\delta_F=\Max_{b\in \Fpn}\delta_F(1,b),
\end{equation}
in this case. This reduction property motivates the notion of locally-APN functions, originally introduced for power functions in even characteristic~\cite{BCC11} and later extended to odd characteristic in~\cite{HLX+23}, where the condition
\begin{equation}\label{locallyapn_def}
	\delta_F(1,b)\le 2\text{ for all }b\in \Fpn\setminus\Fp,
\end{equation}
was introduced.
A further generalization was considered in~\cite{KK26} for functions satisfying
\begin{equation}\label{locallyapn_cordef_eqn}
	\delta_F(a,b)= \delta_F(1,g_a(b))\text{ for all }b\in \Fpn,
\end{equation}
where $g_a$ permutes $\Fpn$ for each $a\in\Fpnmul$. 

\begin{defi}\label{locallyapn_cordef}
	Let $F$ be a function on $\Fpn$ satisfying \eqref{locallyapn_cordef_eqn} for every $a\in \Fpnmul$. Then, \begin{itemize}
		\item $F$ is called \textbf{locally-PN}, if
		$\delta_F(1,b)\le 1$ for all $b\in \Fpn\setminus \Fp$.
		\item $F$ is called \textbf{locally-APN}, if
		$\delta_F(1,b)\le 2$ for all $b\in \Fpn\setminus \Fp$.
	\end{itemize}
\end{defi}

The following lemma shows that $F_{r,u}$ satisfies \eqref{locallyapn_cordef_eqn}, and hence the notion of locally-APN functions applies to $F_{r,u}$ as well.

\begin{lemma}\label{Fruproperty2_lemma}\cite[Lemma 11]{MW25}
	Let $F_{r,u}(x) = x^r (1+u\chi(x))$ be defined on $\Fpn$, where $r>1$ and $u\in \Fpnmul$. Let $a\in \Fpnmul$ and $b\in \Fpn$. Then,
	\begin{align*}
		\delta_{F_{r,u}}(a,b)&=
		\begin{cases}
			\delta_{F_{r,u}}\left( 1, \frac{b}{a^r}\right) &\text{ if }\chi(a)=1,\\
			\delta_{F_{r,u}}\left( 1, \frac{b}{(-1)^{r+1}a^r}\right) &\text{ if }\chi(a)=-1,
		\end{cases}
		\\
		\beta_{F_{r,u}}(a,b)&=
		\begin{cases}
			\beta_{F_{r,u}}\left( 1, \frac{b}{a^r}\right) &\text{ if }\chi(a)=1,\\
			\beta_{F_{r,u}}\left( 1, \frac{b}{(-1)^{r}a^r}\right) &\text{ if }\chi(a)=-1.
		\end{cases}
	\end{align*}
\end{lemma}

\subsection{Known Results on Quadratic Character}\label{subsec_chi}

We denote 
\begin{equation*}
	S_{ij}=\{x\in \Fpn : \chi(x)=(-1)^i, \chi(x+1)=(-1)^j\},
\end{equation*}
where $i,j \in \{0,1\}$. The following lemma is well-known and useful for our results.

\begin{lemma}\label{Sset num of elts}~\cite{Dic35} If $p^n\equiv 1\pmod{4}$, then $\#S_{00}=\frac{p^n-5}{4}$ and $\#S_{01}=\#S_{10}=\#S_{11}=\frac{p^n-1}{4}$. If $p^n\equiv 3\pmod{4}$, then $\#S_{00}=\#S_{10}=\#S_{11}=\frac{p^n-3}{4}$ and $\#S_{01}=\frac{p^n+1}{4}$.
\end{lemma}

The following three lemmas on the quadratic character are useful for our results.

\begin{lemma}\label{chi_lemma quad}\cite[Theorem 5.48]{LN97}
	Let $f(x) = a_2x^2 +a_1 x + a_0 \in \Fpn [x]$ with $p$ odd and $a_2 \ne 0$. Put $d=a_1^2-4a_0a_2$. Then,
	\begin{equation*}
		\sum_{x\in \Fpn} \chi(f(x)) = 
		\begin{cases}
			-\chi(a_2) &\text{ if }d\ne 0,\\
			(p^n-1)\chi(a_2) &\text{ if }d = 0.
		\end{cases}
	\end{equation*}
\end{lemma}

\begin{lemma}\label{chi_lemma inequal}\cite[Theorem 5.41]{LN97}
	Let $\eta(\cdot)$ be a multiplicative character of $\Fpn$ of order $m>1$ and let $f\in\Fpn[x]$ be a monic polynomial of positive degree that is not $m$-th power of a polynomial. Let $d$ be the number of distinct roots of $f$ in its splitting field over $\Fpn$. Then for every $a\in\Fpn$ we have
	\begin{equation*}
		\left| \sum_{x\in \Fpn} \eta (af(x))\right| \le (d-1)\sqrt{p^n}.
	\end{equation*}
\end{lemma}

The following lemma is well-known. A simple proof is given in Lemma 2.5 of \cite{KKKK25}.

\begin{lemma}\label{chi_lemma cubic}\cite[Exercise 5.59]{LN97}
	Let $p^n\equiv 3\pmod{4}$ and $f$ be a function on $\Fpn$ with $f(-x)=-f(x)$ for all $x\in \Fpn$. Then, $\Sum_{x\in \Fpn}\chi(f(x))=0$.
\end{lemma}

\subsection{Differential and Boomerang Properties of $F_r$}\label{subsec_KKKK25}

To compute the differential uniformity of $F_r$, we count the number of solutions of 
\begin{equation}\label{DU_eqn}
	b=F_r(x+1)-F_r(x) = (x+1)^r(1+\chi(x+1))-x^r(1+\chi(x)).
\end{equation}
Let $D_{ij} (b)$ denote the number of solutions of \eqref{DU_eqn} in $S_{ij}$ where $i,j\in\{0,1\}$. We summarize the results of \cite{KKKK25} on differential properties of $F_r$ in the following theorem. 

\begin{theorem}\label{KKKK25 DU_thm}\cite[Section 3.1]{KKKK25}
	Let $r>1$ be an integer and $q=p^n \equiv 3 \pmod{4}$. If $(x+1)^r-x^r =b$ has at most $1$ solution in $S_{00}$ for all $b\in \Fpnmul$ and $\gcd(r,q-1) \mid 2$, then
	\begin{itemize}
		\item \cite[Lemma 3.2]{KKKK25} $D_{11}(0)=\frac{p^n-3}{4}$ and $D_{11}(b)=0$ for all $b\in \Fpnmul$.
		\item \cite[Lemma 3.3]{KKKK25} $D_{00}(0)=0$, and $D_{00}(b)\le 1$ for all $b\in \Fpnmul$.
		\item \cite[Lemma 3.4, 3.5]{KKKK25} $D_{01}(b) + D_{10}(b)=0$ for $b\in \{0,2\}$, and $D_{01}(b) + D_{10}(b) \le 1$ for all $b\in \Fpn\setminus \{0,2\}$.
		\item \cite[Theorem 3.1]{KKKK25} Combining the above results, $\delta_{F_r}(1,b) \le 2$ for all $b\in \Fpnmul$ and
		\begin{equation*}
			\delta_{F_r}=\delta_{F_r}(1,0)=\frac{p^n+1}{4},
		\end{equation*}
		and hence $F_r$ is locally-APN.
	\end{itemize} 
\end{theorem}

To compute the boomerang uniformity of $F_r$, we consider the number of common solutions $(x,y)$ of the following system.
\begin{equation}\label{BU_system}
	\begin{cases}
		x^r(1+\chi(x))-y^r(1+\chi(y))=b,\\
		(x+1)^r(1+\chi(x+1))-(y+1)^r(1+\chi(y+1))=b.
	\end{cases}
\end{equation}
Let $B_{ijkl}(b)$ denote the number of solutions of \eqref{BU_system} in $S_{ij}\times S_{kl}$, where $i,j,k,l\in\{0,1\}$. We summarize the results of \cite{KKKK25} on boomerang properties of $F_r$ in the following theorem. 

\begin{theorem}\label{KKKK25 BU_thm} \cite[Section 4.1]{KKKK25}
	Under the same assumptions as in Theorem \ref{KKKK25 DU_thm}, we have  
	\begin{enumerate}
		\item \cite[Lemma 4.5]{KKKK25} $B_{0001}(b) = B_{0010}(b) =0$ or $B_{0100}(b)=B_{1000}(b)=0$.
		\item \cite[Lemma 4.6]{KKKK25} $B_{0001}(b)$,  $B_{0010}(b)$, $B_{0100}(b)$, $B_{1000}(b)\le 1$.
		\item \cite[Theorem 4.1]{KKKK25} Let $\alpha\in S_{00}$ be the solution of $(x+1)^r-x^r =1$ if it exists. Then, $B_{0001}(b) = B_{0010}(b) = B_{1000}(b)=0$, and hence $\beta_{F_r}(1,\pm 2\alpha^r)=1+ B_{0100}(b) \le 2$.
		\item \cite[Theorem 4.1]{KKKK25} If $b \ne \pm 2\alpha^r$, then  $\beta_{F_r}(1,b) = B_{0001}(b) + B_{0010}(b) + B_{0100}(b) + B_{1000}(b) \le 2$.
	\end{enumerate}
\end{theorem}

For any function $F$ satisfying \eqref{locallyapn_cordef_eqn}, the differential spectrum of $F$ is defined to be the multiset $DS_F = \{\omega_i : 0\le i \le \delta_F\}$, where
\begin{equation*}
	\omega_i = \#\{b\in \Fpn : \delta_F(1,b)=i\}. 
\end{equation*} 
The following identity for the differential spectrum is well-known:
\begin{equation}\label{DS_identity}
	\sum_{i=0}^{\delta_F}\omega_i = \sum_{i=0}^{\delta_F}i\cdot \omega_i = p^n.
\end{equation}

It is known that $F_r$ is locally-PN with boomerang uniformity $0$ in several cases, such as $r=3^n-2$ for $p=3$ \cite{LWZ24}, and $r=\frac{p^n+1}{4}$ for $p^n\equiv 3 \pmod{8}$ \cite{KK26}. 
Motivated by these results, we establish a general condition under which $F_r$ is locally-PN, determine its differential spectrum, and show that its boomerang uniformity is $0$.

\begin{prop}\label{Locally-PN_prop}
	Under the assumptions of Theorem \ref{KKKK25 DU_thm}, suppose that
	\begin{itemize}
		\item $D_{00}(2)=0$,
		\item $D_{00}(b)=1$ and $D_{01}(b)+D_{10}(b)=1$ do not occur simultaneously for all $b\in \Fpnmul$.
	\end{itemize}
	Then, $F_r$ is locally-PN with $\delta_{F_r}=\frac{p^n+1}{4}$ and the differential spectrum of $F_r$ is given by
	\begin{equation}\label{Diff_spec Locally-PN}
		DS_{F_r} = \left\{\omega_0 = \frac{p^n-3}{4},\ \omega_1 = \frac{3p^n-1}{4},\ \omega_{\frac{p^n+1}{4}}=1\right\}.
	\end{equation}
\end{prop}

\begin{proof}
	If $x=0,-1$ in \eqref{DU_eqn}, then $b=2,0$, respectively. 
	By Theorem \ref{KKKK25 DU_thm}, we have $\delta_{F_r}(1,0)=\frac{p^n+1}{4}$, $\delta_{F_r}(1,2)=1+D_{00}(2)$, and 
	$$
	\delta_{F_r}(1,b)\le D_{00}(b)+D_{01}(b)+D_{10}(b).
	$$
	for all $b\ne 0,2$.
	Hence, it suffices to verify the conditions in this proposition to ensure that $\delta_{F_r}(1,b)\le 1$ for all $b\in \Fpnmul$.
	
	Applying $\omega_{\frac{p^n+1}{4}}=1$ on \eqref{DS_identity}, we have $\omega_1=\frac{3p^n-1}{4}$ and $\omega_0 = \frac{p^n-3}{4}$.
\end{proof}

\begin{prop}\label{BU0_prop}
	Let $p^n \equiv 3\pmod{4}$ and $\gcd(r,p^n-1)\in \{1,2\}$. Then, $\beta_{F_r}=0$ if and only if $\delta_{F_r}(1,b)\le 1$ for all $b\in \Fpnmul$.
\end{prop}

\begin{proof}
	Since the argument in the proof of \cite[Theorem 16]{KK26} depends only on the property $\delta_{F_r}(1,b)\le 1$, it can be applied without modification to conclude that $\beta_{F_r}=0$.
	
	Conversely, suppose on the contrary that $\delta_{F_r}(1,b)\ge 2$ for some $b\in \Fpnmul$. Then, there exist $x_1 \ne x_2 \in \Fpn$ such that $F_r (x_1+1) - F_r (x_1)= F_r (x_2+1) - F_r (x_2)= b$. Let $c =F_r(x_1) - F_r (x_2) = F_r(x_1+1) - F_r (x_2+1)$.
	
	\begin{itemize}
		\item If $\chi(x_1)\ne \chi(x_2)$, then one of $F_r(x_1)$ and $F_r(x_2)$ is zero and the other is nonzero, and hence $c\ne 0$.
		\item If $\chi(x_1)= \chi(x_2)=-1$, then we have $F_r(x_1+1) = F_r(x_2+1) =b$. If $\chi(x_1+1)=-1$ or $\chi(x_2 +1)=-1$, then we obtain $b=0$, a contradiction. Hence, we have $\chi(x_1+1)=\chi(x_2+1)=1$. Then, $(x_1+1)^r = (x_2+1)^r$.
		\begin{itemize}
			\item If $\gcd(r,p^n-1)=1$, then we have $x_1=x_2$, a contradiction to $x_1 \ne x_2$.
			\item If $\gcd(r,p^n-1)=2$, then $x_1+1 = -(x_2+1)$, since $x_1 \ne x_2$. But, we can see that $\chi(x_1+1) = \chi(-(x_2+1)) = -\chi(x_2+1)$, which is a contradiction to $\chi(x_1+1)=\chi(x_2+1)=1$. 
		\end{itemize} 
		Hence, this case cannot happen.
		\item Assume that $\chi(x_1)= \chi(x_2)=1$. If $\chi(x_1+1)=-1$ or $\chi(x_2 +1)=-1$ then we can easily see that $c\ne 0$. In the case $\chi(x_1+1)=\chi(x_2+1)=1$, suppose on the contrary that $c=0$. Then, we have $x_1^r = x_2^r$. Since $\gcd(r,p^n-1)\in \{1,2\}$ and $x_1 \ne x_2$, we obtain $x_1 = -x_2$, which contradicts $\chi(x_1)= \chi(x_2)=1$. Hence, $c\ne 0$.
	\end{itemize}
	Thus, in every possible case, we have $c\ne 0$. Therefore, $\beta_{F_r} \ge \beta_{F_r}(1,c)\ge 1$.
\end{proof}

\section{On APN Power Functions from \cite{ZW10} and Their Instances}\label{sec_APN}

The following result is a consequence of the analysis in the proof of Theorem 3.1 together with Theorem 4.1 in \cite{ZW10}. 
We state it explicitly for later use.

\begin{theorem}\cite{ZW10}\label{ZW10_theorem}
	Assume that 
	\begin{equation}\label{ZW10 assumption}
		(3^m+1)r-2 = k(3^n-1), \text{ where }r\text{ is even, }k\text{ is odd, }\gcd(m,n)=1.
	\end{equation}
	\begin{itemize}
		\item Let $b\in \Fcn \setminus \F_3$.
		\begin{itemize}
			\item If $\chi\left( b^{3^m+1}-1\right)=-1$, then $(x+1)^r -x^r =b$ has exactly two solutions, one in $S_{00}\cup S_{11}$ and the other in $S_{01}\cup S_{10}$.
			\item If $\chi\left( b^{3^m+1}-1\right)=1$, then $(x+1)^r -x^r =b$ has no solution in $\Fcn \setminus \{0,-1\}$.
		\end{itemize}
		\item If $b = 0, 1, -1$, then $(x+1)^r -x^r =b$ has one solution $x=1, 0, -1$, respectively.
	\end{itemize}
	Hence, $x^r$ is APN.
\end{theorem}

Note that the differential spectrum of the APN power functions arising from \cite{ZW10} can be derived from Theorem \ref{ZW10_theorem}. 
However, to the best of our knowledge, this spectrum has not been explicitly stated in the literature. 
We therefore summarize the corresponding differential spectrum in the following theorem. We remark that the special case $r = \frac{3^{n+1} -1}{4}$ of this APN class (see Remark~\ref{New APN_corollary remark} (ii)) was recently studied in \cite{XBCH26}, where its differential spectrum was explicitly determined.

\begin{theorem}\label{DS_theorem x^(3^m-1)} 
	If \eqref{ZW10 assumption} holds, then the differential spectrum of $F(x)=x^r$ is given by
	\begin{equation*}
		DS_F = \left\{ 	\omega_0 = \omega_2 = \frac{3^n-3}{2}, \ \ \omega_1= 3 \right\}.
	\end{equation*}
\end{theorem}

\begin{proof}
	From Theorem \ref{ZW10_theorem}, we have
	\begin{equation*}
		\delta_{F}(1,b) =
		\begin{cases}
			2 &\text{ if }\chi\left( b^{3^m+1}-1 \right) =-1, b\in \Fcn \setminus \F_3,\\
			1 &\text{ if }b \in \F_3,\\
			0 &\text{ if }\chi\left( b^{3^m+1}-1 \right) =1, b\in \Fcn \setminus \F_3.
		\end{cases}
	\end{equation*}
	In particular, $\omega_1 =3$. Using $\delta_F=2$ and \eqref{DS_identity}, we obtain the desired differential spectrum.
\end{proof}

We note that condition \eqref{ZW10 assumption} cannot hold when $n$ is even, as can be seen by considering \eqref{ZW10 assumption} modulo $4$.
The following proposition provides a systematic construction of APN power functions arising from the framework of Theorem \ref{ZW10_theorem} for odd $n$. 
In particular, for any integer $m$ with $\gcd(m,n)=1$, it yields an explicit exponent $r$ such that $x^r$ is APN.

\begin{prop}\label{ZW10_general_prop}
	Let $n$ be odd and $\gcd(m,n)=1$. Choose an integer $u$ with $1\le u \le n-1$ such that $um\equiv 1 \pmod{n}$.
	\begin{enumerate}[(i)]
		\item If $u$ is even, then $x^r$, where
		\begin{equation*}
			r= \Frac{3^{um}-1}{3^m+1},
		\end{equation*}
		is APN.
		\item If $u$ is odd, then $x^r$, where
		\begin{equation*}
			r = \Frac{1-3^{(n-u)m}}{1+3^m},
		\end{equation*}
		is APN.
	\end{enumerate}
	Here and throughout the paper, exponents are considered modulo $3^n-1$.
\end{prop}

\begin{proof} Let $t = \frac{um-1}{n}$.\\
	(i) Since $u$ is even, 
	\begin{equation*}
		r = \frac{3^{um}-1}{3^m+1} = 3^{m(u-1)} - 3^{m(u-2)} + \cdots +3^m -1 \equiv u \pmod{2}
	\end{equation*}
	is also even. Moreover, it follows that $nt = um-1$ is odd. Thus, $t$ is also odd, and hence 
	\begin{equation*}
		k = \frac{3(3^{tn}-1)}{3^n-1} = 3 (3^{(t-1)n}+ 3^{(t-2)n}+ \cdots +1) \equiv t \pmod{2}
	\end{equation*}
	is also odd. Moreover, we have
	\begin{equation*}
		(3^m+1)r = 3^{um}-1 = 3^{tn+1} -1 = 3(3^{tn}-1)+2 = k(3^n-1)+2.
	\end{equation*}
	Therefore, \eqref{ZW10 assumption} holds, and hence $x^r$ is APN, by Theorem \ref{ZW10_theorem}.\\
	(ii) Since $u$ is odd, 
	\begin{equation*}
		r = \frac{1-3^{(n-u)m}}{1+3^m} = 1 - 3^m + \cdots + 3^{m(n-u-2)}- 3^{m(n-u-1)}  \equiv n-u \pmod{2}
	\end{equation*}
	is even. Moreover, it follows that $n(m-t)=m(n-u)+1$ is odd. Thus, $m-t=\frac{m(n-u)+1}{n}$ is also odd. Hence, we can see that 
	\begin{equation*}
		k = \frac{1-3^{n(m-t)}}{1-3^n} = 1+ 3^n + \cdots + 3^{(m-t-1)n} \equiv m-t \pmod{2}
	\end{equation*} 
	is odd. Let $d=3r$. Then, we obtain
	\begin{equation*}
		(3^m+1)d = 3(1-3^{(n-u)m}) = (1-3^{n(m-t)}) +2 = -k(3^n-1)+2.
	\end{equation*}
	Therefore, $d$ satisfies \eqref{ZW10 assumption}, and hence $x^d$ is APN, by Theorem \ref{ZW10_theorem}. Therefore, $x^r$ is also APN.
\end{proof}

Proposition~\ref{ZW10_general_prop} provides an explicit construction of an even exponent $r$ satisfying \eqref{ZW10 assumption}. The following result shows that such an $r$ is unique up to linear equivalence (or equivalently, up to cyclotomic cosets).

\begin{prop}\label{ZW10_unique_prop}
	Let $n$ be odd. For each $m$ coprime to $n$, there is a unique even residue $r$ modulo $3^n-1$ satisfying \eqref{ZW10 assumption}.
\end{prop}
\begin{proof}
	The existence follows from Proposition~\ref{ZW10_general_prop}. Let
	\[
	A=\frac{3^m+1}{2},
	\qquad
	B=\frac{3^n-1}{2}.
	\]
	We first claim that $\gcd(A,B)=1$.  Indeed, if an odd prime $\ell$ divided both $A$ and $B$, then
	\[
	3^m \equiv -1 \pmod \ell,
	\qquad
	3^n \equiv 1 \pmod \ell.
	\]
	Therefore the order of $3$ modulo $\ell$ divides both $2m$ and $n$.  Since $n$ is odd and
	$\gcd(m,n)=1$, we have $\gcd(2m,n)=1$, so this order must be $1$.  But then $3 \equiv 1 \pmod \ell$,
	which contradicts $3^m \equiv -1 \pmod \ell$.  Hence $\gcd(A,B)=1$.
	
	Now let $(r_1,k_1)$ and $(r_2,k_2)$ be two solutions of \eqref{ZW10 assumption} with $r_1,r_2$ even and $k_1,k_2$ odd.
	Subtracting the two equations gives
	\[
	A(r_1-r_2)=B(k_1-k_2).
	\]
	Because $\gcd(A,B)=1$, it follows that $B \mid (r_1-r_2)$.  Since $B$ is odd and $r_1-r_2$ is even,
	we actually have
	\[
	2B=3^n-1 \mid (r_1-r_2).
	\]
	Thus $r_1 \equiv r_2 \pmod{3^n-1}$.  In particular, all even solutions give the same residue class,
	hence the same cyclotomic coset.
\end{proof}

\begin{remark}\label{ZW10_general_remark}
	Our exhaustive computational experiments via SageMath indicate that Proposition~\ref{ZW10_general_prop} provides a broad coverage of APN power functions in characteristic $3$. 
	More precisely, for $n \le 13$, among the APN power functions found in our experiments, all those not listed in Table~1 of \cite{BP25} appear to be accounted for by Proposition~\ref{ZW10_general_prop}.
	In particular, Table~\ref{table_ZW10 exponents} in Section \ref{sec_table} illustrates how the APN exponents arising from this construction explain several such cases, while also serving as a source of functions $F_r$ with boomerang uniformity $0$.
\end{remark}

\begin{remark}\label{ZW10_symmetry_remark}
	The following observation shows that the cases $m$ and $n-m$ in Proposition~\ref{ZW10_general_prop} give rise to linearly equivalent APN power functions, so that these two choices of $m$ should be regarded as equivalent.
	
	Let $m_2 = n-m_1$ with $\gcd(m_1,n)=1$, and let $r_i$ be the exponent obtained from Proposition~\ref{ZW10_general_prop} for $m=m_i$, where $i=1,2$.
	Let $u_i$ be the integer such that $1 \le u_i \le n-1$ and $m_i u_i \equiv 1 \pmod{n}$ for $i=1,2$. Since
	$
	(n-u_1)m_2 = (n-u_1)(n-m_1) \equiv u_1m_1 \equiv 1 \pmod{n},
	$
	we have $u_2 = n-u_1$. For convenience, we assume that $u_1$ is even; the case where $u_1$ is odd follows similarly. Then
	\begin{align*}
		3^{m_1(u_1-1)}r_2
		&= 3^{m_1(u_1-1)} \cdot \frac{1-3^{(n-u_2)m_2}}{1+3^{m_2}} = 3^{m_1(u_1-1)} \cdot \frac{1-3^{u_1(n-m_1)}}{1+3^{n-m_1}} \\
		&= 3^{m_1(u_1-1)}
		\left(1-3^{n-m_1}+3^{2(n-m_1)}-\cdots+3^{(u_1-2)(n-m_1)}-3^{(u_1-1)(n-m_1)}\right) \\
		&= 3^{m_1(u_1-1)} - 3^{n+m_1(u_1-2)} + 3^{2n+m_1(u_1-3)} - \cdots + 3^{(u_1-2)n+m_1} - 3^{(u_1-1)n} \\
		&\equiv 3^{m_1(u_1-1)} - 3^{m_1(u_1-2)} + 3^{m_1(u_1-3)} - \cdots + 3^{m_1} - 1 \pmod{3^n-1} \\
		&= \frac{3^{u_1m_1}-1}{3^{m_1}+1}
		= r_1.
	\end{align*}
	Therefore, $x^{r_1}$ and $x^{r_2}$ are linearly equivalent. Hence, it suffices to consider $m \le \frac{n-1}{2}$ when enumerating APN power functions arising from Proposition~\ref{ZW10_general_prop}.
\end{remark}

The following proposition reformulates Proposition~\ref{ZW10_general_prop} in a more explicit way, implying two families of APN exponents depending on divisibility conditions on $n\pm 1$.

\begin{prop}\label{Our_general_prop}
	Let $n$ be odd.
	\begin{enumerate}[(i)]
		\item If $m\mid n+1$ with $\frac{n+1}{m}$ is even and $m\ne 1$, then $x^r$, where
		\begin{equation*}
			r= \frac{3^{n+1}-1}{3^m+1},
		\end{equation*}
		is APN.
		\item If $m\mid n-1$ with $\frac{n-1}{m}$ is even, then $x^r$, where
		\begin{equation*}
			r = \Frac{1-3^{n-1}}{1+3^m},
		\end{equation*}
		is APN.
	\end{enumerate}
\end{prop}

\begin{proof}
	(i)	Observe that $u = \frac{n+1}{m}$ is even. Since $um=n+1 \equiv 1 \pmod{n}$, Proposition~\ref{ZW10_general_prop}(i) implies that if
	$$
	r=\frac{3^{um}-1}{3^m+1}=\frac{3^{n+1}-1}{3^m+1},
	$$
	then $x^r$ is APN.\\
	(ii) Let $u'=\frac{n-1}{m}$. Observe that $u = n-u' = n-\frac{n-1}{m}$ is odd. Since $um=n(m-1)+1 \equiv 1 \pmod{n}$, Proposition~\ref{ZW10_general_prop}(ii) implies that if
	$$
	r=\frac{1-3^{(n-u)m}}{1+3^m}=\frac{1-3^{u'm}}{1+3^m}=\frac{1-3^{n-1}}{1+3^m},
	$$
	then $x^r$ is APN.
\end{proof}

The following APN class obtained in \cite{Led12} can be recovered as a special case of Proposition~\ref{Our_general_prop} (i).

\begin{corollary}\cite{Led12}\label{Led12_corollary}
	Let $n \equiv -1 \pmod{2^\ell}$ and $m=\frac{n+1}{2^\ell}$ where $\ell$ is a positive integer. Then, the function $x^r$, where
	$$
	r = \frac{3^{n+1}-1}{3^m+1} = \frac{3^{n+1}-1}{3^{\frac{n+1}{2^\ell}}+1},
	$$
	is APN.
\end{corollary}

\begin{proof}
	Since $n=2^\ell m-1$, we have $m\mid n+1$ and $\frac{n+1}{m}=2^\ell$ is even. 
	If $m\ne 1$, then the result follows from Proposition~\ref{Our_general_prop}(i). If $m=1$, then $r=\frac{3^{n+1}-1}{4}$, which is the known APN exponent from \cite{HRS99}; see also Remark \ref{New APN_corollary remark} (ii).
\end{proof}

By symmetry, one may also consider the case where $\frac{n-1}{m}$ is a power of $2$. 
This leads to the following APN class.

\begin{corollary}\label{New APN_corollary}
	Let $n \equiv 1 \pmod{2^\ell}$ and $m=\frac{n-1}{2^\ell}$ where $\ell$ is a positive integer. Then the function $x^r$, where
	$$
	r = \frac{1-3^{n-1}}{1+3^m} = \frac{1-3^{n-1}}{1+3^{\frac{n-1}{2^\ell}}},
	$$
	is APN.
\end{corollary}

\begin{proof}
	Since $n=2^\ell m+1$, we have $m\mid n-1$ and $\frac{n-1}{m}=2^\ell$ is even. 
	Thus the result follows from Proposition~\ref{Our_general_prop}(ii).
\end{proof}

\begin{remark}\label{New APN_corollary remark}
	\begin{enumerate}[(i)]
		\item It is easy to see that if $\ell =1$ in Corollary \ref{Led12_corollary}, we obtain $r = 3^{\frac{n+1}{2}} -1$ which is a known APN exponent from \cite{HRS99}. Moreover, taking $\ell =1$ in Corollary~\ref{New APN_corollary} yields
		$$
		 r = 1 - 3^{\frac{n-1}{2}} \equiv 3^n - 3^{\frac{n-1}{2}} = 3^{\frac{n-1}{2}}(3^{\frac{n+1}{2}} - 1) \pmod{3^n-1},
		$$
		which is equivalent to the same class.
		\item Taking $m=1$ in Proposition~\ref{Our_general_prop} (ii) gives
		$$
		3r \equiv 3r+(3^n-1) =  \frac{3^{n+1}-1}{4} \pmod{3^n-1},
		$$
		which is also a known APN exponent from \cite{HRS99}.
		\item Outside of these cases, the construction in Corollary \ref{New APN_corollary} produces new instances starting from $n=9$.
	\end{enumerate}
\end{remark}

\section{Locally-PN Binomials with Boomerang Uniformity $0$}\label{sec_locally-PN}

In this section we introduce two classes of locally-PN binomials $F_r$ with boomerang uniformity $0$,
using Proposition \ref{Locally-PN_prop} and Proposition \ref{BU0_prop}.

\subsection{The APN-exponent class arising from \eqref{ZW10 assumption}}\label{subsec_r=ZW10}

In this subsection, we show that the corresponding function $F_r$ is locally-PN with boomerang uniformity $0$, where $r$ is an APN exponent arising from \cite{ZW10} satisfying \eqref{ZW10 assumption}. 
Such exponents were discussed in Section~\ref{sec_APN}. We remark that, for the associated function $F_r$, the case $m=1$ (i.e., $r=\frac{3^{n+1}-1}{4}$) was already studied in \cite{KK26}. 

By Theorem~\ref{ZW10_theorem}, it remains to verify the second condition of Proposition~\ref{Locally-PN_prop}. We consider the number of solutions of the following equation for every $b\in \Fcnmul$. 
\begin{equation}\label{DU_equation r=ZW10}
	(x+1)^{r}(1+\chi(x+1))-x^{r}(1+\chi(x))=b
\end{equation}

\begin{lemma}\label{DU_lemma r=ZW10}
	Assume that \eqref{ZW10 assumption} holds. If $b\in \Fcnmul$, then
	\begin{align*}
		D_{01}(b)=1\ &\Leftrightarrow\ \chi(b)=1,\ \chi\left( b^{\frac{3^m+1}{2}}+1\right) =-1,\\
		D_{10}(b)=1\ &\Leftrightarrow\ \chi(b)=-1,\ \chi\left( b^{\frac{3^m+1}{2}}+(-1)^m\right) =(-1)^m.
	\end{align*}
\end{lemma}
\begin{proof}
	If $x\in S_{01}$, then \eqref{DU_equation r=ZW10} leads to
	\begin{equation*}
		x^{r} = b.
	\end{equation*}
	Since $\chi(x)=1$, we have $\chi(b) = \chi\left( x^{r}\right)  = 1$. 
	Raising both sides of the above equation to the $\frac{3^{m}+1}{2}$-th power implies 
	\begin{equation*}
		b^{\frac{3^{m}+1}{2}} =x^{r(3^m+1)}= x^{\frac{k(3^{n}-1)}{2}+1} = \chi(x^k)\cdot x =x,
	\end{equation*}
	since $\chi(x)=1$. Hence, $\chi(x+1)=-1$ if and only if $\chi\left( b^{\frac{3^m+1}{2}}+1\right)=-1$. Hence, $D_{01}(b)=1$ if and only if $\chi(b)=1$ and $\chi\left( b^{\frac{3^m+1}{2}}+1\right)=-1$.
	
	If $x\in S_{10}$, then \eqref{DU_equation r=ZW10} leads to
	\begin{equation*}
		(x+1)^{r} = -b.
	\end{equation*}
	Since $\chi(x+1)=1$, we have $\chi(b) = \chi\left( -(x+1)^{r}\right)  = -1$. 
	Raising both sides of the above equation to the $\frac{3^m+1}{2}$-th power implies 
	\begin{equation*}
		(-b)^{\frac{3^m+1}{2}} = (x+1)^{\frac{k(3^{n}-1)}{2}+1} = x+1,
	\end{equation*}
	since $\chi(x+1)=1$. Hence, $\chi(x)=-1$ if and only if $\chi\left( (-b)^{\frac{3^m+1}{2}}-1\right)=-1$, or equivalently $\chi\left( b^{\frac{3^m+1}{2}} +(-1)^m\right) = (-1)^m$. 
\end{proof}

\begin{theorem}\label{DU_theorem r=3^m-1}
	Assume that \eqref{ZW10 assumption} holds. Then, $F_{r}$ is locally-PN with the differential uniformity
	$$\delta_{F_{r}}=\delta_{F_{r}}(1,0) = \frac{3^n+1}{4},$$
	and the differential spectrum is given in \eqref{Diff_spec Locally-PN}. Moreover, $F_{r}$ has boomerang uniformity $0$.
\end{theorem}
\begin{proof}
	As mentioned in the beginning of this subsection, it suffices to prove that if $D_{00}(b)=1$ then $D_{01}(b) = 0$ and $D_{10}(b) = 0$ for every $b\in \Fcnmul$.  
	
	Suppose that \eqref{DU_equation r=ZW10} has a solution $x=x_0$, where $x_0 \in S_{00}$. Then,
	\begin{equation*}
		-b = (x_0+1)^{r} - x_0^{r} = \left( (x_0+1)^{\frac{r}{2}} - x_0^{\frac{r}{2}}\right) \left( (x_0+1)^{\frac{r}{2}} + x_0^{\frac{r}{2}}\right). 
	\end{equation*} 
	Let $t= (x_0+1)^{\frac{r}{2}} - x_0^{\frac{r}{2}}$. Then, $-\frac{b}{t} = (x_0+1)^{\frac{r}{2}} + x_0^{\frac{r}{2}}$, and hence
	\begin{equation*}
		(x_0 +1)^{\frac{r}{2}}  = \frac{b}{t}-t, \ \ x_0^{\frac{r}{2}} = \frac{b}{t}+t.
	\end{equation*}
	Raising $(3^m+1)$-th power on the second equation, we have
	\begin{equation*}
		\left( \frac{b}{t}+t \right)^{3^m+1}  = x_0^{\frac{(3^m+1)r}{2}}= x_0^{\frac{k(3^{n}-1)}{2}+1} = \chi(x_0^k)\cdot x_0 = x_0,
	\end{equation*}
	because $x_0 \in S_{00}$. Similarly, $\left( \frac{b}{t}-t \right)^{3^m+1}  = x_0+1$. Thus,
	\begin{align}
		1&= \left( \frac{b}{t}-t \right)^{3^m+1} - \left( \frac{b}{t}+t \right)^{3^m+1} \notag \\
		&= \left( \frac{b^{3^m+1}}{t^{3^m+1}} - \frac{b^{3^m}}{t^{3^m-1}} - bt^{3^m-1} +t^{3^m+1}\right) 
		- \left( \frac{b^{3^m+1}}{t^{3^m+1}} + \frac{b^{3^m}}{t^{3^m-1}} + bt^{3^m-1} +t^{3^m+1}\right) \notag \\
		&=b\left( t^{3^m-1}+\frac{b^{3^m-1}}{t^{3^m-1}}\right)  = b\left( t^{3^m-1}+\frac{b^{3^m-1}}{t^{3^m-1}} + b^{\frac{3^m-1}{2}} - b^{\frac{3^m-1}{2}}\right) \notag \\
		&=b\left( t^{\frac{3^m-1}{2}} - \frac{b^{\frac{3^m-1}{2}}}{t^{\frac{3^m-1}{2}}} \right) ^2 - b^{\frac{3^m+1}{2}} \label{DU_eqn 01 r=(3^2m-1)(3^m+1)}\\
		&=b\left( t^{\frac{3^m-1}{2}} + \frac{b^{\frac{3^m-1}{2}}}{t^{\frac{3^m-1}{2}}} \right) ^2 + b^{\frac{3^m+1}{2}}. \label{DU_eqn 10 r=(3^2m-1)(3^m+1)}
	\end{align}
	\eqref{DU_eqn 01 r=(3^2m-1)(3^m+1)} and \eqref{DU_eqn 10 r=(3^2m-1)(3^m+1)} imply $\chi(b)=\chi\left( b^{\frac{3^m+1}{2}}+1\right)$ and $\chi(b)\ne\chi\left( b^{\frac{3^m+1}{2}}-1\right)$, respectively. By Lemma \ref{DU_lemma r=ZW10}, we obtain $D_{01}(b)=0$ and $D_{10}(b)=0$. This completes the proof by Propositions~\ref{Locally-PN_prop} and \ref{BU0_prop}.
\end{proof}

\subsection{The case $r=2\cdot 3^{\frac{n-1}{2}} +1$}\label{subsec_r=2*3^l+1}

In this section, we study differential and boomerang properties of $F_{2\cdot 3^\ell + 1}$ on $\Fcn$, when $n$ is odd and $\ell=\frac{n-1}{2}$. First, we extract several auxiliary results from \cite{DHKM01}, which arise in the derivation of the differential spectrum of the power function $x^{2\cdot 3^\ell + 1}$, and present them as a lemma for later use.

\begin{lemma}\label{DHKM01_lemma}\cite{DHKM01}
	Let $n$ be odd and 
	$$
	g(x) = (x-1)^{2\cdot 3^\ell + 1} - (x+1)^{2\cdot 3^\ell + 1}
	$$
	where $\ell=\frac{n-1}{2}$. Let 
	\begin{align*}
		S = \{x \in \Fcnmul : \chi(x^2-1)=-1\},& \ \  S' = \{x \in \Fcnmul : \chi(x^2-1)=1\},\\
		P = \{x \in \Fcn : \chi(x^{3-2\cdot 3^{\ell+1}}-1)=1\},& \ \ P' = \{x \in \Fcn : \chi(x^{3-2\cdot 3^{\ell+1}}-1)=-1\}.
	\end{align*}
	Then,
	\begin{enumerate}
		\item If $x\in S$, then $g(x)\in P$. If $x\in S'$, then $g(x) \in P'$. If $x\in \F_3$, then $g(x)=1$.
		\item The restriction of $g$ to $S$ maps 4-to-1, and the restriction of $g$ to $S'$ maps 2-to-1.
	\end{enumerate}
\end{lemma}

To study differential properties of $F_r$, we count the number of solutions of the following equation for every $b\in \Fcnmul$. 
\begin{equation}\label{DU_equation r=2*3^l+1}
	(x+1)^{2\cdot 3^{\frac{n-1}{2}}+1} (1+\chi(x+1))-x^{2\cdot 3^{\frac{n-1}{2}}+1}(1+\chi(x))=b
\end{equation}

\begin{lemma}\label{DU_lemma S00 r=2*3^l+1}
	Let $n$ be odd and $\ell=\frac{n-1}{2}$. Then, $D_{00}(b) \le 1$. Moreover, $D_{00}(b)=1$ if and only if $\chi\left(b\left(  b^{2\cdot 3^{\ell+1} -3}+1\right) \right)=1$.
\end{lemma}
\begin{proof}
	If \eqref{DU_equation r=2*3^l+1} has a solution in $S_{00}$, then it is also a solution of 
	\begin{equation}\label{DU_eqn1 00 r=2*3^l+1}
		-b = (x+1)^{2\cdot 3^\ell + 1} - x^{2\cdot 3^\ell + 1} = g(x-1)
	\end{equation}
	where $g$ is given in Lemma \ref{DHKM01_lemma}. Observe that $x-1 \in S'$ if and only if $\chi(x)\chi(x+1)=1$. Hence, by Lemma \ref{DHKM01_lemma}, $-b\in P'$, and hence 
	\begin{align*}
		-1&= \chi\left( \left( -b\right)^{3-2\cdot 3^{\ell+1}} -1 \right)  = -\chi\left(  b^{3-2\cdot 3^{\ell+1}} +1 \right)  = - \chi\left( b^{3-2\cdot 3^{\ell+1}}\right) \chi\left( b^{2\cdot 3^{\ell+1} -3}+1 \right)  \\
		& = - \chi(b)\chi\left( b^{2\cdot 3^{\ell+1} -3}+1 \right),
	\end{align*}
	as desired. 
	
	Conversely, if $\chi\left(b\left(  b^{2\cdot 3^{\ell+1} -3}+1\right) \right)=1$, then $-b \in P'$. By Lemma \ref{DHKM01_lemma}, \eqref{DU_eqn1 00 r=2*3^l+1} has two solutions $x_1, x_2$ such that $\chi(x_1(x_1+1)) = \chi(x_2(x_2+1))=1$. Moreover, since $r=2\cdot 3^\ell + 1$ is odd, if $x=x_1$ is a solution of \eqref{DU_eqn1 00 r=2*3^l+1}, then $x=-x_1-1$ is also a solution of \eqref{DU_eqn1 00 r=2*3^l+1}. Thus, we can see that $x_2 = -x_1 -1$. Observe that if $x_1 \in S_{00}$ then $\chi(x_2) = \chi(-(x_1+1)) =-1$ and $\chi(x_2+1) = \chi(-x_1) = -1$, so $x_2 \in S_{11}$. Therefore, we conclude that \eqref{DU_eqn1 00 r=2*3^l+1} has exactly one solution in $S_{00}$ and hence $D_{00}(b)=1$.
\end{proof}

\begin{lemma}\label{DU_lemma r=2*3^l+1}
	Let $n$ be odd and $\ell=\frac{n-1}{2}$. If $b\in \Fcnmul$, then
	\begin{align*}
		D_{01}(b)=1\ &\Leftrightarrow\ \chi(b)=1,\ \chi\left( b^{2\cdot 3^{\ell+1} -3}+1\right) =-1,\\
		D_{10}(b)=1\ &\Leftrightarrow\ \chi(b)=-1,\ \chi\left( b^{2\cdot 3^{\ell+1} -3}+1\right) =1.
	\end{align*}
	Moreover, $D_{01}(b) + D_{10}(b) = 1$ if and only if $\chi\left(b\left(  b^{2\cdot 3^{\ell+1} -3}+1\right) \right)=-1$.
\end{lemma}
\begin{proof}
	If $x\in S_{01}$, then \eqref{DU_equation r=2*3^l+1} leads to
	\begin{equation*}
		x^{2\cdot 3^\ell + 1} = b
	\end{equation*}
	Since $\chi(x)=1$, we have $\chi(b) = \chi\left( x^{2\cdot 3^\ell + 1}\right)  = 1$. 
	Raising both sides of the above equation to the $(2\cdot 3^\ell - 1)$-th power implies 
	\begin{equation*}
		b^{2\cdot 3^\ell - 1}  = x^{4\cdot 3^{n-1}-1} = x^{(3^n-1)+3^{n-1}} =x^{3^{n-1}},
	\end{equation*}
	and hence $x= \left(b^{2\cdot 3^\ell - 1} \right)^3 = b^{2\cdot 3^{\ell+1} -3}$.  Hence, $\chi(x+1)=-1$ if and only if $\chi\left( b^{2\cdot 3^{\ell+1} -3}+1\right)=-1$. Hence, $D_{01}(b)=1$ if and only if $\chi(b)=1$ and $\chi\left( b^{2\cdot 3^{\ell+1} -3}+1\right)=-1$.
	
	The proof of the characterization of $D_{10}(b)$ is analogous, and we omit it here. The final statement follows immediately from these two equivalences.
\end{proof}

\begin{theorem}\label{DU_theorem r=2*3^l+1}
	Let $n$ be odd and $\ell=\frac{n-1}{2}$. Then, $F_{2\cdot 3^\ell+1}$ is locally-PN with the differential uniformity
	$$\delta_{F_{2\cdot 3^\ell+1}}=\delta_{F_{2\cdot 3^\ell+1}}(1,0) = \frac{3^n+1}{4},$$
	and the differential spectrum is given in \eqref{Diff_spec Locally-PN}. Moreover, $F_{2\cdot 3^\ell+1}$ has boomerang uniformity $0$.
\end{theorem}

\begin{proof}
	We apply Proposition \ref{Locally-PN_prop}. 
	Since $3^n-1=(2\cdot 3^\ell+1)(2\cdot 3^\ell-1)-3^{n-1}$,
	we obtain $\gcd(2\cdot 3^\ell+1,3^n-1)=1$. 
	Next, by Lemmas \ref{DU_lemma S00 r=2*3^l+1} and \ref{DU_lemma r=2*3^l+1}, $D_{00}(b)\le 1$ for all $b\in \Fcnmul$, and the conditions $D_{00}(b)=1$ and $D_{01}(b)+D_{10}(b)=1$ do not occur simultaneously for any $b\in \Fcn\setminus\{0,2\}$. 
	
	Now, it remains to show that $D_{00}(2)=D_{00}(-1)=0$. 
	Substituting $b=-1$ into \eqref{DU_eqn1 00 r=2*3^l+1}, we obtain $g(x-1)=1$. 
	By Lemma \ref{DHKM01_lemma}, all solutions of $g(x-1)=1$ lie in $\F_3$. 
	Since $S_{00}\cap \F_3=\emptyset$, it follows that $D_{00}(-1)=0$.
\end{proof}

\section{Differential and Boomerang Properties of $F_r$ with $r=3^n-3$}\label{sec_r=-2}

In this section, we study differential and boomerang properties of $F_{3^n-3}$ on $\Fcn$, when $n$ is odd. In the following theorem, we show that $F_{3^n-3}$ is locally-APN, using Theorem \ref{KKKK25 DU_thm}.

\begin{theorem}\label{DU_theorem r=-2}
	Let $n\ge 3$ be odd. Then, $\delta_{F_{3^n-3}}(1,b)\le 2$ for all $b\in \Fcnmul$ and $\delta_{F_{3^n-3}}(1,0)= \frac{3^n+1}{4}$. Therefore, $F_{3^n-3}$ is locally-APN.
\end{theorem}

\begin{proof}
	By Theorem \ref{KKKK25 DU_thm}, it suffices to show that  
	\begin{equation}\label{DU_equation r=-2}
		(x+1)^{3^n-3} -x^{3^n-3} = -b.
	\end{equation}
	has at most one solution in $S_{00}$. Suppose  that $x\not \in \{0,-1\}$. Then, we have $x^2 - (x+1)^2 = -bx^2(x+1)^2$, or equivalently
	\begin{equation*}
		x^4 -x^3 + x^2 +\frac{1}{b}x-\frac{1}{b}=0.
	\end{equation*}
	Applying $x=y+1$ in the above equation, we have
	\begin{equation}\label{DU_eqn1 r=-2}
		y^4 + y^2 -uy +1=0,
	\end{equation}
	where $u = -\frac{1}{b}$. If \eqref{DU_eqn1 r=-2} has two or more solutions, 
	\begin{equation*}
		g(y) = y^4 + y^2 -uy +1 = (y^2 +cy +d)(y^2 -cy +d^{-1}) 
	\end{equation*}
	for some $c, d\in \Fcnmul$. Then,
	\begin{equation}\label{DU_system r=-2}
		\begin{cases}
			d+d^{-1} = c^2 +1,\\
			c\left( d-d^{-1}\right) =u
		\end{cases}
	\end{equation}
	The first equation in \eqref{DU_system r=-2} leads to
	\begin{align*}
		d^2+1 &= d(c^2+1)\\
		d^2+2d(c^2+1)+1&=0\\
		d^2+2d(c^2+1)+c^4-c^2+1&=c^4-c^2\\
		\left(d+(c^2+1)\right)^2 &= c^2(c^2-1).
	\end{align*}
	If $c^2=0$ or $c^2=1$, then the above equation implies $d=-(c^2+1)$ which leads to $d^{-1} = (c^2+1)-d = d$ and hence $u=0$, a contradiction. Hence, $c^2(c^2-1)\ne 0$ and $\chi(c^2-1)=1$. Furthermore, we have 
	\begin{equation}
		d = -(c^2+1) + \epsilon c(c^2 -1)^{\frac{3^n+1}{4}} = \left( c-\epsilon(c^2 -1)^{\frac{3^n+1}{4}}\right) ^2
	\end{equation}
	where $\epsilon \in \{\pm1\}$ will be determined later. Then,
	\begin{align*}
		u &= c\left( d-d^{-1}\right) = c \left( -(c^2+1) + \epsilon c(c^2 -1)^{\frac{3^n+1}{4}} -\Frac{1}{-(c^2+1) + \epsilon c(c^2 -1)^{\frac{3^n+1}{4}}}\right)\\
		&= c \left( -(c^2+1) + \epsilon c(c^2 -1)^{\frac{3^n+1}{4}} -\Frac{-(c^2+1) - \epsilon c(c^2 -1)^{\frac{3^n+1}{4}}}{(c^2+1)^2 -  c^2(c^2 -1)}\right)\\
		&= - \epsilon c^2(c^2 -1)^{\frac{3^n+1}{4}}.
	\end{align*}
	Furthermore, we have
	\begin{align*}
		c^2 - d &= c^2 - \left( -(c^2+1) + \epsilon c(c^2 -1)^{\frac{3^n+1}{4}}\right)  = -(c^2-1) - \epsilon c(c^2 -1)^{\frac{3^n+1}{4}}\\
		&=(c^2 -1)^{\frac{3^n+1}{4}}\left( 2 c\epsilon -(c^2 -1)^{\frac{3^n+1}{4}}\right).
	\end{align*}
	Similarly, we obtain
	\begin{equation*}
		c^2 - d^{-1} =-(c^2 -1)^{\frac{3^n+1}{4}}\left( 2 c \epsilon+(c^2 -1)^{\frac{3^n+1}{4}}\right). 
	\end{equation*}
		 
	
	If $\epsilon =\chi(c+1) = \chi(c-1)$, then 
	\begin{equation*}
		c^2 -d = (c^2 -1)^{\frac{3^n+1}{4}} \left( (c +1)^{\frac{3^n+1}{4}} +(c -1)^{\frac{3^n+1}{4}}\right)^2
	\end{equation*}
	is a square, and 
	\begin{equation*}
		c^2 - d^{-1} = - (c^2 -1)^{\frac{3^n+1}{4}} \left( (c +1)^{\frac{3^n+1}{4}} -(c -1)^{\frac{3^n+1}{4}}\right)^2
	\end{equation*}
	is not a square. Hence, $g(y)=0$ has two solutions satisfying $y^2+cy+d=0$. Let $y_1$ and $y_2$ be two solutions of $y^2+cy+d=0$. Then, $x_1 = y_1+1$ and $x_2 = y_2 +1$ are two solutions of \eqref{DU_equation r=-2}, and
	\begin{align*}
		x_1 x_2 (x_1+1)(x_2+1) &= (y_1+1)(y_2+1)(y_1-1)(y_2-1) = (y_1^2-1)(y_2^2-1) = y_1^2y_2^2 - (y_1^2+y_2^2)+1 \\
		&=(y_1y_2)^2 -(y_1+y_2)^2 - y_1y_2 +1  = d^2 - c^2-d+1 = (d+1)^2 -c^2 \\
		&= \left( -c^2 +\epsilon c(c^2-1)^{\frac{3^n+1}{4}} \right) ^2 -c^2 = c^4+c^2(c^2-1) +\epsilon c^3(c^2-1)^{\frac{3^n+1}{4}}-c^2 \\
		&=-c^2 \left( c^2-1 - \epsilon c(c^2-1)^{\frac{3^n+1}{4}}\right)  = -c^2(c^2-1)^{\frac{3^n+1}{4}}\left(2\epsilon c +(c^2-1)^{\frac{3^n+1}{4}} \right) \\
		&=-c^2(c^2-1)^{\frac{3^n+1}{4}}\left( (c+1)^{\frac{3^n+1}{4}} +(c-1)^{\frac{3^n+1}{4}}\right)^2.
	\end{align*}
	Hence, we have $\chi\left(x_1  (x_1+1)\right) \chi\left( x_2 (x_2+1)\right) = \chi\left( x_1 x_2 (x_1+1)(x_2+1)\right) =-1$, and hence satisfying both $\chi\left(x_1  (x_1+1)\right)=1$ and $\chi\left(x_2  (x_2+1)\right)=1$ is impossible. Therefore, \eqref{DU_equation r=-2} has at most one solution in $S_{00}$.
	
	The proof for the case $-\epsilon =\chi(c+1) = \chi(c-1)$ is very similar to the above case, and we omit it here.
\end{proof}

Next, we study the boomerang properties of $F_{3^n-3}$. In the following lemma, we give necessary and sufficient conditions for each $B_{ijkl}(b)$ to be equal to $1$.

\begin{lemma}\label{BS_lemma r=-2}
	Let $p=3$, $n\ge 3$ be odd, $r=3^n-3$ and $b\in \Fcnmul$. Then,
	\begin{itemize}
		\item $B_{0001}(b)=1$ if and only if 
		\begin{equation}\label{BU_eqn r=-2 0001}
			\chi(b)=-1, \chi\left(b^{\frac{3^n+1}{4}} + 1 \right) =\chi\left(b^{\frac{3^n+1}{4}} + b \right) =1, \chi\left( b^{\frac{3^n+1}{4}}\left( b^{\frac{3^n+1}{4}}+b\right)^{\frac{3^n+1}{4}} - (b^{\frac{3^n+1}{4}}+1) \right)=1.
		\end{equation}
		\item $B_{0010}(b)=1$ if and only if 
		\begin{equation}\label{BU_eqn r=-2 0010}
			\chi(b)
			= \chi\!\left(b^{\frac{3^n+1}{4}} - 1\right)=
			\chi\!\left(b^{\frac{3^n+1}{4}} - b\right)=-1, \chi\left( b^{\frac{3^n+1}{4}}\left( b^{\frac{3^n+1}{4}}-b\right)^{\frac{3^n+1}{4}} + (b^{\frac{3^n+1}{4}}-1) \right)=1.
		\end{equation}
		\item $B_{0100}(b)=1$ if and only if 
		\begin{equation}\label{BU_eqn r=-2 0100}
			\chi(b)= 1,\ 
			\chi\!\left(b^{\frac{3^n+1}{4}} + b\right)=
			\chi\!\left(b^{\frac{3^n+1}{4}} - 1\right) = -1, \chi\left( b^{\frac{3^n+1}{4}}\left( b^{\frac{3^n+1}{4}}+b\right)^{\frac{3^n+1}{4}} + (b^{\frac{3^n+1}{4}}-1) \right)=1.
		\end{equation}
		\item $B_{1000}(b)=1$ if and only if 
		\begin{equation}\label{BU_eqn r=-2 1000}
			\chi(b)
			=\chi\!\left(b^{\frac{3^n+1}{4}} + 1\right)
			=\chi\!\left(b^{\frac{3^n+1}{4}} - b\right) 
			=1, \chi\left( b^{\frac{3^n+1}{4}}\left( b^{\frac{3^n+1}{4}}-b\right)^{\frac{3^n+1}{4}} - (b^{\frac{3^n+1}{4}}+1) \right)=1.
		\end{equation}
	\end{itemize} 
\end{lemma}

\begin{proof}
	If there is a solution $(x,y)\in S_{00}\times S_{01}$ in \eqref{BU_system} with $r=3^n-3$, then \eqref{BU_system} reduces to 
	\begin{equation}\label{BU_system S0001 r=-2}
		\begin{cases}
			x^{3^n-3}-y^{3^n-3} = -b,\\
			(x+1)^{3^n-3} = -b.
		\end{cases}
	\end{equation} 
	Raising both sides of the second equation in \eqref{BU_system S0001 r=-2} to the $\frac{3^n+1}{4}$-th power yields $x + 1 = -\Frac{1}{b^{\frac{3^n+1}{4}}}$ and $\chi(b) = -1$, 
	since $\frac{3^n+1}{4}$ is odd when $n$ is odd and $\chi(x + 1) = 1$.
	Then, $x = -\Frac{1}{b^{\frac{3^n+1}{4}}} - 1 = -  \Frac{b^{\frac{3^n+1}{4}}+1} {b^{\frac{3^n+1}{4}}}$, hence $\chi(x)=1$ leads to $\chi(b^{\frac{3^n+1}{4}}+1)=1$. Substitute $x = -  \Frac{b^{\frac{3^n+1}{4}}+1} {b^{\frac{3^n+1}{4}}}$ into the first equation in \eqref{BU_system S0001 r=-2}, we obtain
	\begin{align*}
		y^{3^n-3} &= x^{3^n-3} + b = \left(-\Frac{b^{\frac{3^n+1}{4}}} {b^{\frac{3^n+1}{4}}+1}\right)^2 + b =\Frac{-b + b\left( b^{\frac{3^n+1}{4}}+1\right)^2}{\left( b^{\frac{3^n+1}{4}}+1\right)^2 } \\
		&=-\Frac{-b + b\left( b^{\frac{3^n+1}{2}} - b^{\frac{3^n+1}{4}} +1\right)}{\left( b^{\frac{3^n+1}{4}}+1\right)^2 }
		= -\Frac{ b\left( b^{\frac{3^n+1}{4}}+b\right)}{\left( b^{\frac{3^n+1}{4}}+1\right)^2 }. 
	\end{align*}
	Hence, the above equation implies
	$$
	y = -\Frac{b^{\frac{3^n+1}{4}}+1 }{ b^{\frac{3^n+1}{4}}\left( b^{\frac{3^n+1}{4}}+b\right)^{\frac{3^n+1}{4}}}.
	$$
	When $\chi(b) = -1$ and $\chi(b^{\frac{3^n+1}{4}}+1)=1$, $\chi(y)=1$ leads to $\chi\left( b^{\frac{3^n+1}{4}}+b\right)=1$. Furthermore, $\chi(y+1)=-1$ is equivalent to 
	$$
	\chi\left( b^{\frac{3^n+1}{4}}\left( b^{\frac{3^n+1}{4}}+b\right)^{\frac{3^n+1}{4}} - (b^{\frac{3^n+1}{4}}+1) \right)=1. 
	$$

	The proofs for the other cases are very similar, and we omit here.	
\end{proof}

The following two lemmas are useful for our results.

\begin{lemma}\label{chi cubic lemma}
	Let $n$ be odd. Then,
	\begin{equation*}
		\Sum_{x\in\Fcn}\chi(x^4+x^3-1)=-1.
	\end{equation*}
\end{lemma}
\begin{proof}
	Let $y= \frac{1}{x}$ where $x\in \Fcnmul$. Then, we obtain
	\begin{align*}
		\Sum_{x\in \Fcn}\chi\left(x^4+x^3-1 \right) 
		&= -1 + \Sum_{x\in \Fcnmul}\chi\left(x^4+x^3-1 \right) 
		= -1 + \Sum_{y\in \Fcnmul}\chi\left(\Frac 1 {y^4}+\Frac{1}{y^3}-1 \right)\\
		&= -1 + \Sum_{y\in \Fcnmul}\chi\left(\Frac {1+y-y^4} {y^4} \right) 
		= -2 - \Sum_{y\in \Fcn}\chi\left(y^4-y-1 \right).
	\end{align*}
	If we set $z=x+1$, then
	\begin{align*}
		\Sum_{x\in \Fcn}\chi\left(x^4+x^3-1 \right) = \Sum_{z\in \Fcn}\chi\left(z^4-z-1 \right) = -2-\Sum_{z\in \Fcn}\chi\left(z^4+z^3-1 \right),
	\end{align*}
	which completes the proof.
\end{proof}

\begin{lemma}\label{chi hyperelliptic lemma}
	Let $n$ be odd. Then,
	\begin{equation*}
		\Sum_{x\in\Fcn}\chi(x(x^2+1))\chi(x^4+x^3-1) + \Sum_{x\in\Fcn}\chi(x^4-1)\chi(x^4+x^3-1)=-1.
	\end{equation*}
\end{lemma}
\begin{proof}
	Let $y=x-1$ and 
	\begin{equation*}
		A = \Sum_{x\in\Fcn}\chi(x^2+1)\chi(x^4+x^3-1)\left( \chi(x) + \chi(x^2-1)\right). 
	\end{equation*}
	Then, 
	\begin{equation*}
		A = \Sum_{y\in\Fcn}\chi(y^2-y-1)\chi(y^4 -y^3 +y+1)\left( \chi(y+1) + \chi(y^2-y)\right). 
	\end{equation*}
	If $z=\frac{1}{y}$, then 
	\begin{align*}
		A &= \Sum_{y\in\Fcn}\chi(y^2-y-1)\chi(y^4 -y^3 +y+1)\left( \chi(y+1) + \chi(y^2-y)\right)\\
		&=-1+\Sum_{y\in\Fcnmul}\chi(y^2-y-1)\chi(y^4 -y^3 +y+1)\left( \chi(y+1) + \chi(y^2-y)\right)\\
		&=-1+\Sum_{z\in\Fcnmul}\chi\left( \frac{1}{z^2}-\frac{1}{z}-1\right) \chi\left( \frac{1}{z^4} -\frac{1}{z^3} +\frac{1}{z}+1\right) \left( \chi\left(\frac{1}{z}+1\right) + \chi\left( \frac{1}{z^2}-\frac{1}{z}\right) \right)\\
		&=-1-\Sum_{z\in\Fcnmul}\chi\left(z^2+z-1\right) \chi\left( z^4+z^3-z+1 \right) \left( \chi\left(z^2+z\right) + \chi\left( 1-z\right) \right)\\
		&=-2-\Sum_{z\in\Fcn}\chi\left(z^2+z-1\right) \chi\left( z^4+z^3-z+1 \right) \left( \chi\left(z^2+z\right) + \chi\left( 1-z\right) \right).
	\end{align*}
	If we set $w=-z+1$, then 
	\begin{equation*}
		A =-2-\Sum_{w\in\Fcn}\chi\left(w^2+1\right) \chi\left( w^4+w^3-1 \right) \left( \chi\left(w^2-1\right) + \chi\left( w\right) \right) = -2-A,
	\end{equation*}
	which completes the proof.
\end{proof}

Applying Lemma \ref{Fruproperty2_lemma}, the boomerang spectrum of $F$ is defined to be the multiset $BS_F=\{\nu_i : 0 \le i \le \beta_F\}$, where
\begin{equation*}
	\nu_i = \#\{ b\in \Fpnmul : \beta_F(1,b)=i\}.
\end{equation*}
The following identity for the boomerang spectrum is well-known:
\begin{equation}\label{BS_identity}
	\sum_{i=0}^{\beta_F}\nu_i = q-1.
\end{equation}

\begin{theorem}\label{BS_theorem}
	If $n \ge 3$ odd, the boomerang spectrum of $F_{3^n-3}$ is given by
	\begin{equation*}
		BS_{F_{3^n-3}} = \left\{ \nu_0 = \frac{1}{4}\left( 3^{n+1} -5-  2\Gamma_1 - \Gamma_2 \right),\ \  \nu_1 = \frac{1}{4}\left( 3^n +1+ 2\Gamma_1 + \Gamma_2 \right) \right\},
	\end{equation*}
	where
	\begin{equation*}
		\Gamma_1 = \Sum_{u\in \Fcn} \chi(u(u^2+1))\chi\left(u^4+u^3-1 \right),\ \
		\Gamma_2 = \Sum_{u\in \Fcn} \chi(u(1-u^2))\chi\left(u^4+u^3-1 \right). 
	\end{equation*}
	Moreover, $\nu_1 > 0$ and hence $\beta_{F_{3^n-3}} =1$, when $n\ge 5$.
\end{theorem}

\begin{proof}
	We first show that $\beta_{F_{3^n-3}}\le 1$. 
	By Theorem \ref{KKKK25 BU_thm}, it suffices to prove that 
	\eqref{BU_eqn r=-2 0001} and \eqref{BU_eqn r=-2 0010} 
	do not hold simultaneously, and that 
	\eqref{BU_eqn r=-2 0100} and \eqref{BU_eqn r=-2 1000} 
	do not hold simultaneously.
	
	Suppose that \eqref{BU_eqn r=-2 0001} and \eqref{BU_eqn r=-2 0010} occur simultaneously. Then, $\chi(b)=-1$. Moreover, 
	\begin{equation*}
		-1= \chi\left( b^{\frac{3^n+1}{4}}+1\right)  \chi\left( b^{\frac{3^n+1}{4}}-1\right)  = \chi\left( b^{\frac{3^n+1}{2}} - 1\right)  = \chi (-b-1),
	\end{equation*}
	and hence $\chi(b+1)=1$. Then, 
	\begin{equation*}
		-1= \chi\left( b^{\frac{3^n+1}{4}}+b\right)  \chi\left( b^{\frac{3^n+1}{4}}-b\right)  = \chi\left( b^{\frac{3^n+1}{2}} - b^2\right)  = \chi(-b-b^2 )= -\chi(b)\chi(b+1)=1,
	\end{equation*}
	a contradiction.	
	The argument for 
	\eqref{BU_eqn r=-2 0100} and \eqref{BU_eqn r=-2 1000} 
	is analogous and hence omitted.
	
	Next, we compute $\nu_1$. \eqref{BU_eqn r=-2 0001} and \eqref{BU_eqn r=-2 0010} are identical after replacing $b$ in \eqref{BU_eqn r=-2 0100} and \eqref{BU_eqn r=-2 1000} by $-b$, respectively. Hence, $\nu_1$ equals twice the number of $b\in \Fcnmul$ satisfying \eqref{BU_eqn r=-2 0001} or \eqref{BU_eqn r=-2 0010}.
	
	Let $A_1$ be the number of $b\in \Fcnmul$ satisfying \eqref{BU_eqn r=-2 0001}. Then,
	\begin{align*}
		A_1&=\frac{1}{16}\Sum_{b\in \Fcn} (1-\chi(b)) \left(1+\chi\left(b^{\frac{3^n+1}{4}}+1\right)\right) \left(1+\chi\left(b^{\frac{3^n+1}{4}}+b\right)\right) \left(1+\chi\left( b^{\frac{3^n+1}{4}}\left( b^{\frac{3^n+1}{4}}+b\right)^{\frac{3^n+1}{4}} - (b^{\frac{3^n+1}{4}}+1) \right)\right)\\
		&=\frac{1}{8}\Sum_{\chi(b)=-1} \left(1+\chi\left(b^{\frac{3^n+1}{4}}+1\right)\right) \left(1+\chi\left(b^{\frac{3^n+1}{4}}+b\right)\right) \left(1+\chi\left( b^{\frac{3^n+1}{4}}\left( b^{\frac{3^n+1}{4}}+b\right)^{\frac{3^n+1}{4}} - (b^{\frac{3^n+1}{4}}+1) \right)\right).
	\end{align*}
	Let $y=b^{\frac{3^n+1}{4}}$. Then, $b=-y^2$ when $\chi(b)=-1$, and hence 
	\begin{align}
		A_1&=\frac{1}{16}\Sum_{y\in \Fcn}  (1+\chi(y+1)) \left(1+\chi\left(y-y^2\right)\right) \left(1 + \chi\left( y\left( y-y^2\right)^{\frac{3^n+1}{4}} - (y+1) \right)\right)\label{A1} \\
		&=\frac{1}{16}\Sum_{y\in \Fcn}  (1+\chi(y+1)) \left(1+\chi\left(y-y^2\right)\right) \notag \\
		&+\frac{1}{16}\Sum_{y\in \Fcn}  (1+\chi(y+1)) \left(1+\chi\left(y-y^2\right)\right)\chi\left( y\left( y-y^2\right)^{\frac{3^n+1}{4}} - (y+1) \right)\notag \\
		&=\frac{1}{16}(3^n+1) + \frac{1}{16}\Sum_{y\in \Fcn}  (1+\chi(y+1)) \left(1+\chi\left(y-y^2\right)\right)\chi\left( y\left( y-y^2\right)^{\frac{3^n+1}{4}} - (y+1) \right), \notag
	\end{align}
	by Lemmas \ref{chi_lemma quad} and \ref{chi_lemma cubic}. Let $t=y+1$. Then,
	\begin{align*}
		&\Sum_{y\in \Fcn}  (1+\chi(y+1)) \left(1+\chi\left(y-y^2\right)\right)\chi\left( y\left( y-y^2\right)^{\frac{3^n+1}{4}} - (y+1) \right) \\
		&= \Sum_{t\in \Fcn}  (1+\chi(t)) \left(1+\chi\left(1-t^2\right)\right)\chi\left( (t-1)\left( 1-t^2\right)^{\frac{3^n+1}{4}} - t \right) \\
		&= -4+\Sum_{t\in \Fcn \setminus \F_3}  (1+\chi(t)) \left(1+\chi\left(1-t^2\right)\right)\chi\left( (t-1)\left( 1-t^2\right)^{\frac{3^n+1}{4}} - t \right)
	\end{align*}
	Let $t = \frac{u}{1+u^2}$. Then,
	\begin{align*}
		&\Sum_{t\in \Fcn \setminus \F_3}  (1+\chi(t)) \left(1+\chi\left(1-t^2\right)\right)\chi\left( (t-1)\left( 1-t^2\right)^{\frac{3^n+1}{4}} - t \right) \\
		&=\frac{1}{2}\Sum_{u\in \Fcn \setminus \F_3}  \left(1+\chi\left(\Frac{u}{u^2+1}\right)\right) \left(1+\chi\left(\Frac{(1-u^2)^2}{(1+u^2)^2}\right)\right)\chi\left( -\Frac{(u+1)^2}{1+u^2}\left( \Frac{1-u^2}{1+u^2}\right)^{\frac{3^n+1}{2}} - \Frac{u}{1+u^2} \right)\\
		&=\frac{1}{2}\Sum_{u\in \Fcn \setminus \F_3}  \left(1+\chi\left(u(u^2+1)\right)\right) \left(1+\chi\left((1-u^4)^2\right)\right)\chi\left(- \Frac{(u+1)^2(1-u^2)\chi(1-u^4)+u(1+u^2)}{(1+u^2)^2}\right)\\
		&=-\Sum_{u\in \Fcn\setminus \F_3}  \left(1+\chi\left(u(u^2+1)\right)\right) \chi\left( (u+1)^3(1-u)\chi(1-u^4)+u(u^2+1)\right)\\
		&=\Sum_{\substack{u\in \Fcn\setminus \F_3\\ \chi(1-u^4)=1}}  \left(1+\chi\left(u(u^2+1)\right)\right) \chi\left(u^4+u^3-1 \right) - \Sum_{\substack{u\in \Fcn\setminus \F_3\\ \chi(1-u^4)=-1}}  \left(1+\chi\left(u(u^2+1)\right)\right) \chi\left(u^4-u-1 \right)\\
		&=\Sum_{u\in \Fcn\setminus \F_3}  (1+\chi(1-u^4))\left(1+\chi\left(u(u^2+1)\right)\right) \chi\left(u^4+u^3-1 \right) 
	\end{align*}
	where the last equality is from
	\begin{equation*}
		\Sum_{\substack{u\in \Fcn\setminus \F_3\\ \chi(1-u^4)=-1}}  \left(1+\chi\left(u(u^2+1)\right)\right) \chi\left(u^4-u-1 \right) = -\Sum_{\substack{s\in \Fcn\setminus \F_3\\ \chi(1-s^4)=1}}  \left(1+\chi\left(s(s^2+1)\right)\right) \chi\left(s^4+s^3-1 \right)
	\end{equation*}
	when $s=\frac{1}{u}$. Hence, by Lemma \ref{chi cubic lemma} and Lemma \ref{chi hyperelliptic lemma}
	\begin{align*}
		A_1 &= \frac{1}{16}\left( 3^n +1 + \Sum_{y\in \Fcn}  (1+\chi(y+1)) \left(1+\chi\left(y-y^2\right)\right)\chi\left( y\left( y-y^2\right)^{\frac{3^n+1}{4}} - (y+1) \right) \right)\\
		&= \frac{1}{16}\left( 3^n -3  +\Sum_{t\in \Fcn \setminus \F_3}  (1+\chi(t)) \left(1+\chi\left(1-t^2\right)\right)\chi\left( (t-1)\left( 1-t^2\right)^{\frac{3^n+1}{4}} - t \right) \right)\\
		&= \frac{1}{16}\left( 3^n -3 + \Sum_{u\in \Fcn\setminus \F_3}  (1+\chi(1-u^4))\left(1+\chi\left(u(u^2+1)\right)\right) \chi\left(u^4+u^3-1 \right) \right)\\
		&= \frac{1}{16}\left( 3^n +1 + \Sum_{u\in \Fcn}  (1+\chi(1-u^4))\left(1+\chi\left(u(u^2+1)\right)\right) \chi\left(u^4+u^3-1 \right) \right)\\
		&= \frac{1}{16} (3^n +1 +2\Gamma_1 + \Gamma_2 ) 
	\end{align*}

	Let $A_2$ be the number of $b\in \Fcnmul$ satisfying \eqref{BU_eqn r=-2 0010}. Then,
	\begin{align*}
		A_2&=\frac{1}{16}\Sum_{b\in \Fcn} (1-\chi(b)) \left(1-\chi\left(b^{\frac{3^n+1}{4}}-1\right)\right) \left(1-\chi\left(b^{\frac{3^n+1}{4}}-b\right)\right) \left(1+\chi\left( b^{\frac{3^n+1}{4}}\left( b^{\frac{3^n+1}{4}}-b\right)^{\frac{3^n+1}{4}} + (b^{\frac{3^n+1}{4}}-1) \right)\right)\\
		&=\frac{1}{8}\Sum_{\chi(b)=-1} \left(1-\chi\left(b^{\frac{3^n+1}{4}}-1\right)\right) \left(1-\chi\left(b^{\frac{3^n+1}{4}}-b\right)\right) \left(1+\chi\left( b^{\frac{3^n+1}{4}}\left( b^{\frac{3^n+1}{4}}-b\right)^{\frac{3^n+1}{4}} + (b^{\frac{3^n+1}{4}}-1) \right)\right).
	\end{align*}
	Let $y=b^{\frac{3^n+1}{4}}$. Then $b=-y^2$ when $\chi(b)=-1$ and hence 
	\begin{align*}
		A_2&=\frac{1}{16}\Sum_{y\in \Fcn}  (1-\chi(y-1)) \left(1-\chi\left(y+y^2\right)\right) \left(1 + \chi\left( y\left( y+y^2\right)^{\frac{3^n+1}{4}} + (y-1) \right)\right)\\
		&=\frac{1}{16}\Sum_{y\in \Fcn}  (1-\chi(-t-1)) \left(1-\chi\left(t^2-t\right)\right) \left(1 + \chi\left( -t\left( t^2-t\right)^{\frac{3^n+1}{4}} - (t+1) \right)\right) = A_1,
	\end{align*}
	where $t=-y$, and by \eqref{A1}.
	
	Therefore, we finally obtain that 
	\begin{equation*}
		\nu_1 = 4A_1 = \frac{1}{4}(3^n +1 + 2\Gamma_1 + \Gamma_2).
	\end{equation*}
	Applying \eqref{BS_identity}, we get the desired boomerang spectrum. 
	
	By Lemma \ref{chi_lemma inequal}, we have $|\Gamma_1| \le 6 \sqrt{3^n}$ and $|\Gamma_2| \le 6 \sqrt{3^n}$. Hence,
	\begin{equation*}
		\nu_1 \ge \frac{1}{4}\left( 3^n + 1 -2 |\Gamma_1| - |\Gamma_2|\right) \ge \frac{1}{4}\left( 3^n - 18 \cdot 3^{\frac{n}{2}}+1\right).
	\end{equation*}
	A direct computation using SageMath shows that $\frac{1}{4}\left( 3^n - 18 \cdot 3^{\frac{n}{2}}\right)\ge 1$ if $3^n\ge 330 \approx 3^{5.28}$. We confirm that $\beta_{F_{3^n-3}}=0$ when $n=3$, and $\beta_{F_{3^n-3}}=1$ when $n=5$. 
\end{proof}

Table \ref{table BS} shows the boomerang spectrum of $F_{3^n-3}$ given in Theorem \ref{BS_theorem} for $3\le n \le 15$ using SageMath. We also verified via SageMath that the results of Theorem \ref{BS_theorem} are correct, when $3\le n \le 9$.
\begin{table}[htbp]
	\centering
	\begin{tabular}{c c c c}
		\toprule
		$n$ & $\Gamma_1$ & $\Gamma_2$ & $BS_{F_{3^n-3}}$ \\
		\midrule
		3   & $-2$ & $-24$  & $\{v_0=26,\; v_1=0\}$ \\
		5   & $-22$ & $40$  & $\{v_0=182,\; v_1=60\}$ \\
		7   & $250$ & $112$ & $\{v_0=1486,\; v_1=700\}$ \\
		9   & $142$ & $48$  & $\{v_0=14678,\; v_1=5004\}$ \\
		11  & $-1586$  & $792$ & $\{v_0=133454,\; v_1=43692\}$ \\
		13  & $570$ & $-104$ & $\{v_0=1195482,\; v_1=398840\}$ \\
		15  & $-6262$ & $-1184$ & $\{v_0=10765106,\; v_1=3583800\}$ \\
		\bottomrule
	\end{tabular}
	\caption{Boomerang spectrum $BS_{F_{3^n-3}}$ when $n$ is odd, $3 \le n \le 15$.}\label{table BS}
\end{table}

\section{Numerical Results in Characteristic $3$}\label{sec_table}

In Section~\ref{subsec_KKKK25}, we presented the results of \cite{KKKK25}, which characterize power functions $F_r$ with boomerang uniformity at most $2$. 
However, as observed in \cite{LWZ24, KK26, KKKK25}, functions $F_r$ with smaller boomerang uniformity, namely $0$ or $1$, appear more frequently in the case $p=3$ than for other characteristics. 
Several of the experimentally observed cases are theoretically explained in the present paper.
Motivated by this phenomenon, we conducted a computational study of exponents $r$ for which $F_r$ has boomerang uniformity $0$ or $1$ when $p=3$. 
Analogously to the classification of APN power functions in \cite{HRS99}, we summarize the exponents $r$ for which $F_r$ has boomerang uniformity $0$ in Table \ref{table_BU0}, and $F_r$ has boomerang uniformity $1$ in Table \ref{table_BU1}. 
In contrast to the case of boomerang uniformity $0$, determining whether $\beta_{F_r}=1$ requires computing the full BCT, which is computationally much more demanding. Therefore, Table~\ref{table_BU1} is restricted to $n \le 7$.

It is easy to see that if $r_i = rp^i$ and $L_i (x) = x^{p^i}$ where $0< i < n$, then 
\begin{equation*}
	(F_{r} \circ L_i) (x)= (x^{p^i})^r (1+\chi(x^{p^i})) = x^{rp^i}(1+\chi(x)) = F_{r_i}(x),
\end{equation*} 
and hence $F_{r_i}$ is linearly equivalent to $F_r$. Moreover, it is obvious that $F_r  = F_{r+\frac{p^n-1}{2}}$. Therefore, in the tables of this section, we describe cyclotomic cosets modulo $\frac{p^n-1}{2}$ of each class.

\begin{table}[htbp]
	\centering
	\begin{tabular}{c c c c c c}
		\toprule
		\textbf{$n$} & \textbf{$r$} &  Cyclotomic Cosets & $\Max_{b\in \Fcnmul}\delta_{F_r}(1,b)$ & \makecell{Algebraic\\ Degree} & Ref. \\
		\midrule
		\multirow{2}{*}{$3$}
		& 7 & $(7,8,11)$ & 1 & 4 & \makecell{\cite{KK26},\\ Sec. \ref{subsec_r=ZW10}, \ref{subsec_r=2*3^l+1}} \\
		&  12 & $(12,10,4)$ & 1 &  5 & \cite{LWZ24} \\
		\midrule
		\multirow{4}{*}{$5$}
		& 19 & $(19, 57, 50, 29, 87)$ & 1 & 6 & Sec. \ref{subsec_r=2*3^l+1} \\
		& 26 & $(26, 78, 113, 97, 49)$ & 1 & 6 & Sec. \ref{subsec_r=ZW10} \\
		& 61 & $(61, 62, 65, 74, 101)$ & 1 & 6 & \cite{KK26}\\ 
		&  120 & $(120,118,112,94,40)$ & 1 & 9 & \cite{LWZ24} \\
		\midrule
		\multirow{5}{*}{$7$}
		& 55 & $(55, 165, 495, 392, 83, 249, 747)$ & 1 & 8 & Sec. \ref{subsec_r=2*3^l+1} \\
		& 80 & $(80, 240, 720, 1067, 1015, 859, 391)$ & 1 & 8 & Sec. \ref{subsec_r=ZW10} \\
		& 547 & $(547, 548, 551, 560, 587, 668, 911)$ & 1 & 8 & \cite{KK26}\\ 
		& 656 & $(656, 875, 439, 224, 672, 923, 583)$ & 1 & 8 & Sec. \ref{subsec_r=ZW10} \\
		& 1092 & $(1092, 1090, 1084, 1066, 1012, 850, 364)$ & 1 & 13 & \cite{LWZ24} \\
		\midrule
		\multirow{5}{*}{$9$}
		& 163 & $(163, 489, 1467, 4401, 3362, 245, 735, 2205, 6615)$ & 1 &  10 & Sec. \ref{subsec_r=2*3^l+1} \\
		& 242 & $(242, 726, 2178, 6534, 9761, 9601, 9121, 7681, 3361)$  & 1 & 10 & Sec. \ref{subsec_r=ZW10} \\
		& 4921 & $(4921, 4922, 4925, 4934, 4961, 5042, 5285, 6014, 8201)$  & 1 & 10 & \cite{KK26}\\ 
		& 9185 & $(9185, 7873, 3937, 1970, 5910, 7889, 3985, 2114, 6342)$  & 1 & 10 & Sec. \ref{subsec_r=ZW10}\\
		& 9840 & $(9840, 9838, 9832, 9814, 9760, 9598, 9112, 7654, 3280)$ & 1  & 17 & \cite{LWZ24} \\
		\bottomrule
	\end{tabular}
	\caption{$F_r$ with boomerang uniformity $0$ on $\Fcn$ when $n \le 9$.}
	\label{table_BU0}
\end{table}

\begin{table}[htbp]
	\centering
	\begin{tabular}{c c c c c c}
		\toprule
		$n$ & $r$ &  Cyclotomic Cosets & $\Max_{b\in \Fcnmul}\delta_{F_r}(1,b)$ & \makecell{Algebraic\\ Degree} & Ref. \\
		\midrule
		$3$	& 2 & $(2, 6, 5)$ & 2 & 3 & \cite{KKKK25}\\
		\midrule
		\multirow{9}{*}{$5$}
		& 2 & $(2, 6, 18, 54, 41)$ & 2 & 5 & \cite{KKKK25} \\
		& 8 & $(8, 24, 72, 95, 43)$ & 2 & 5 & \\
		& 10 & $(10, 30, 90, 28, 84)$ & 2 & 7 & \\
		& 13 & $(13, 39, 117, 109, 85)$ & 2 & 8 & \\ 
		& 16 & $(16, 48, 23, 69, 86)$ & 3 & 5 & \\
		& 20 & $(20, 60, 59, 56, 47)$ & 2 & 5 & \\
		& 31 & $(31, 93, 37, 111, 91)$ & 2 & 8 & \\ 
		& 67 & $(67, 80, 119, 115, 103)$ & 2 & 8 & Sec. \ref{sec_r=-2} \\
		& 76 & $(76, 107, 79, 116, 106)$ & 2 & 7 & \\
		\midrule
		\multirow{6}{*}{$7$}
		& 2 & $(2, 6, 18, 54, 162, 486, 365)$ & 2 & 7 & \cite{KKKK25} \\
		& 5 & $(5, 15, 45, 135, 405, 122, 366)$ & 2 & 6 & \\
		& 107 & $(107, 321, 963, 703, 1016, 862, 400)$ & 2 & 7 & \\
		& 169 & $(169, 507, 428, 191, 573, 626, 785)$ & 3 & 8 & \\ 
		& 182 & $(182, 546, 545, 542, 533, 506, 425)$ & 2 & 7 & \\ 
		& 1091 & $(1091, 1087, 1075, 1039, 931, 607, 728)$ & 2 & 12 & Sec. \ref{sec_r=-2} \\
		\bottomrule
	\end{tabular}
	\caption{$F_r$ with boomerang uniformity $1$ on $\Fcn$ when $n\le 7$.}
	\label{table_BU1}
\end{table}

For $n=11, 13$, due to computational limitations, we instead searched for exponents $r$ with $\Max_{b\in \Fcnmul}\delta_{F_r}(1,b)=1$, which is described in Table \ref{table_Locally-PN n=11,13}.
Our computational results, obtained while constructing Tables \ref{table_BU0} and \ref{table_BU1}, suggest that, for $n\le 9$, the condition $\delta_{F_r}(1,b)\le 1$ for all $b\in \Fcnmul$ is a necessary and sufficient condition for $F_r$ to have boomerang uniformity $0$, even without the assumption $\gcd(r,p^n-1)\in\{1,2\}$ in Proposition \ref{BU0_prop}. 
Therefore, we also expect that Table \ref{table_Locally-PN n=11,13} provides a complete list of exponents $r$ such that $F_r$ has boomerang uniformity $0$ when $n=11,13$. Finally, we remark that we have completed rigorous proofs for all functions identified with boomerang uniformity $0$ in our computational results.

\begin{table}[htbp]
	\centering
	\scriptsize
	\begin{tabular}{c c c c c}
		\toprule
		$n$ & $r$ &  Cyclotomic Cosets  & \makecell{Algebraic\\ Degree} & Ref. \\
		\midrule
		\multirow{7}{*}{$11$}		
		& 487 & \makecell{(487, 1461, 4383, 13149, 39447, 29768,\\ 731, 2193, 6579, 19737, 59211)} &  12 & Sec. \ref{subsec_r=2*3^l+1} \\
		& 728 & \makecell{(728, 2184, 6552, 19656, 58968, 88331,\\ 87847, 86395, 82039, 68971, 29767)} & 12 & Sec. \ref{subsec_r=ZW10} \\
		& 18980&\makecell{(18980, 56940, 82247, 69595, 31639, 6344,\\ 19032, 57096, 82715, 70999, 35851)} & 12 & Sec. \ref{subsec_r=ZW10}\\ 
		& 44287 & \makecell{(44287, 44288, 44291, 44300, 44327, 44408,\\ 44651, 45380, 47567, 54128, 73811)} & 12 & \cite{KK26}\\
		& 53144 & \makecell{(53144, 70859, 35431, 17720, 53160, 70907,\\ 35575, 18152, 54456, 74795, 47239)} & 12 & Sec. \ref{subsec_r=ZW10}\\
		& 74891 & \makecell{(74891, 47527, 54008, 73451, 43207, 41048,\\ 34571, 15140, 45420, 47687, 54488)} & 12 & Sec. \ref{subsec_r=ZW10}\\
		& 88572 & \makecell{(88572, 88570, 88564, 88546, 88492, 88330,\\ 87844, 86386, 82012, 68890, 29524)} & 21 & \cite{LWZ24} \\
		\midrule
		\multirow{8}{*}{$13$}
		& 1459 & \makecell{(1459, 4377, 13131, 39393, 118179, 354537,\\ 266450, 2189, 6567, 19701, 59103, 177309, 531927)} &  14 & Sec. \ref{subsec_r=2*3^l+1} \\
		& 2186 & \makecell{(2186, 6558, 19674, 59022, 177066, 531198, 796433,\\ 794977, 790609, 777505, 738193, 620257, 266449)} & 14 & Sec. \ref{subsec_r=ZW10} \\
		& 398581 & \makecell{(398581, 398582, 398585, 398594, 398621, 398702,\\ 398945, 399674, 401861, 408422, 428105, 487154, 664301)} & 14 & \cite{KK26}\\
		& 408302 &\makecell{(408302, 427745, 486074, 661061, 388861, 369422,\\ 311105, 136154, 408462, 428225, 487514, 665381, 401821)} & 14 & Sec. \ref{subsec_r=ZW10}\\ 
		& 490058 &\makecell{(490058, 673013, 424717, 476990, 633809, 307105,\\ 124154, 372462, 320225, 163514, 490542, 674465, 429073)} & 14 & Sec. \ref{subsec_r=ZW10}\\ 
		& 744017 & \makecell{(744017, 637729, 318865, 159434, 478302, 637745,\\ 318913, 159578, 478734, 639041, 322801, 171242, 513726)} & 14 & Sec. \ref{subsec_r=ZW10}\\
		& 778181 & \makecell{(778181, 740221, 626341, 284701, 56942, 170826,\\ 512478, 740273, 626497, 285169, 58346, 175038, 525114)} & 14 & Sec. \ref{subsec_r=ZW10}\\
		& 797160 & \makecell{(797160, 797158, 797152, 797134, 797080, 796918,\\ 796432, 794974, 790600, 777478, 738112, 620014, 265720)} & 25 & \cite{LWZ24} \\
		\bottomrule
	\end{tabular}
	\caption{$F_r$ with $\Max_{b\in \Fcnmul}\delta_{F_r}(1,b)=1$ when $n=11,13$.}
	\label{table_Locally-PN n=11,13}
\end{table}

Table~\ref{table_ZW10 exponents} provides a detailed correspondence between the exponents $r$ identified in our numerical search and the theoretical APN classes analyzed in Section~\ref{sec_APN}. This table serves not only to categorize the binomials $F_r$ with boomerang uniformity 0 as discussed in Section~\ref{subsec_r=ZW10}, but also to validate the effectiveness of our generalized parametrization for APN power functions in Section~\ref{sec_APN}. Notably, for larger fields such as $n=11$ and $n=13$, several exponents are not covered by the more specific constructions in Corollaries~\ref{Led12_corollary} and \ref{New APN_corollary}, yet they are successfully accounted for by the broader framework of Proposition~\ref{ZW10_general_prop}.
For example, when $n=11$, the exponent $r=74891$ arises from Proposition~\ref{ZW10_general_prop}(ii) with $(m,u,n)=(4,3,11)$, which yields
$$
r = \frac{1-3^{(n-u)m}}{1+3^m}
= \frac{1-3^{32}}{1+3^4}
= -22597807181120 \equiv 163464
= 74891 + \frac{3^n-1}{2} \pmod{3^n-1}.
$$

\renewcommand{\arraystretch}{1.1}
\begin{table}
	\centering
	\begin{tabular}{c c c}
		\toprule
		$n$ & $r$ &   Source \\
		\midrule
		5 &  26 & Remark \ref{New APN_corollary remark} (i)\\
		7 &  80 & Remark \ref{New APN_corollary remark} (i)\\
		7 &  656 & $\ell=2$ in Corollary \ref{Led12_corollary}\\
		9 &  242 & Remark \ref{New APN_corollary remark} (i)\\
		9 & 9185 & $\ell=2$ in Corollary \ref{New APN_corollary}\\
		11 & 728 & Remark \ref{New APN_corollary remark} (i)\\
		11 & 18980 & $\ell=2$ in Corollary \ref{Led12_corollary}\\
		11 & 53144 & $m=2$, $u=6$ in Proposition \ref{ZW10_general_prop} (i)\\
		11 & 74891 & $m=4$, $u=3$ in Proposition \ref{ZW10_general_prop} (ii) \\
		13 & 2186 & Remark \ref{New APN_corollary remark} (i)\\
		13 & 408302 & $m=4$, $u=10$ in Proposition \ref{ZW10_general_prop} (i)\\
		13 & 490058 & $m=5$, $u=8$ in Proposition \ref{ZW10_general_prop} (i)\\
		13 & 744017 & $m=2$, $u=7$ in Proposition \ref{ZW10_general_prop} (ii)\\
		13 & 778181 & $\ell=2$ in Corollary \ref{New APN_corollary}\\
		\bottomrule
	\end{tabular}
	\caption{Exponents $r$ yielding binomials $F_r$ with boomerang uniformity $0$ in Section~\ref{subsec_r=ZW10}, together with their realization within APN exponent constructions described in Section~\ref{sec_APN}.}
	\label{table_ZW10 exponents}
\end{table}

\section{Conclusion}\label{sec_con}

In this paper, we investigated differential and boomerang properties of binomial functions of the form $F_r(x)=x^r(1+\chi(x))$
over finite fields of characteristic $3$. Our main focus was on identifying exponents $r$ for which $F_r$ exhibits very low boomerang uniformity.
We showed that $F_r$ attains boomerang uniformity $0$ for two classes of exponents, namely those arising from APN exponents in~\cite{ZW10} and the class $r=2\cdot 3^{\frac{n-1}{2}}+1$. We also proved that $F_r$ has boomerang uniformity $1$ when $r=3^n-3$ for $n\ge 5$ odd, and determined its boomerang spectrum. These results demonstrate that, in characteristic $3$, the boomerang uniformity of $F_r$ can be strictly smaller than the general bound obtained in~\cite{KKKK25}.
In addition, we provided a detailed analysis of APN exponents from~\cite{ZW10}. We established an explicit parametrization of such exponents (Proposition \ref{ZW10_general_prop}) and showed that, in characteristic $3$, this construction accounts for all APN exponents arising from~\cite{ZW10}. Our computational results further indicate that, for $n\le 13$, all APN power functions not listed in Table~1 of~\cite{BP25} can be explained by this parametrization.
Moreover, our numerical results suggest that, in characteristic $3$, the condition for a binomial $F_r$ to be locally-PN is both necessary and sufficient for it to have boomerang uniformity $0$.
Finally, we conducted an exhaustive search for small values of $n$, which supports our theoretical findings and illustrates the distribution of exponents yielding low boomerang uniformity.

The results of this paper suggest several directions for future research. 
First, our numerical results indicate that the locally-PN property appears to be closely related to boomerang uniformity $0$ for binomials $F_r$. Establishing a precise equivalence between these two properties remains an interesting open problem.
Second, while our exhaustive search identifies all exponents $r$ yielding boomerang uniformity $0$ or $1$ for small values of $n$, several cases with boomerang uniformity $1$ appearing in Table~\ref{table_BU1} are not yet explained by infinite classes. Extending these sporadic examples to infinite families would be a natural continuation of this work.

\bigskip
\noindent\textbf{Acknowledgments}: This work was supported by the National Research Foundation
of Korea (NRF) grant funded by the Korea government (MSIT) (No. RS-2021-NR061794). Soonhak Kwon was supported by Basic Science Research Program through the National Research Foundation of Korea (NRF) funded by the Ministry of Education (No. RS-2019-NR040081).

\end{document}